\documentclass[12pt, reqno]{amsart}
\pdfoutput=1

\makeatletter
\g@addto@macro{\endabstract}{\@setabstract}
\makeatother

\makeatletter
\def\@seccntformat#1{%
	\ifstrequal{#1}{subsection}
	{\bfseries\csname the#1\endcsname.}
	{\csname the#1\endcsname.} 
}
\makeatother

\usepackage{manyfoot}
\DeclareNewFootnote{default}  
\DeclareNewFootnote{symbol}[fnsymbol]

\usepackage{graphics, stackrel}
\usepackage{amsmath, amssymb, amsthm}
\usepackage{graphicx}
\usepackage{verbatim}
\usepackage{amsfonts}
\usepackage[round,authoryear]{natbib}
\usepackage{fancyvrb}
\usepackage{color}

\usepackage[citecolor=blue, colorlinks=true, linkcolor=blue]{hyperref}

\usepackage{mathrsfs}
\usepackage{bbm}

\usepackage[left=1.25in, right=1.25in, top=1.25in, bottom=1.25in, includehead, includefoot]{geometry}

\usepackage{multirow}
\usepackage{algorithm}
\usepackage{enumitem}

\usepackage{booktabs}
\usepackage{diagbox}
\usepackage{makecell}
\usepackage{rotating}
\usepackage{adjustbox}
\usepackage[graphicx]{realboxes}
\usepackage{lscape}
\usepackage{hhline}
\newcolumntype{C}[1]{>{\centering\arraybackslash}p{#1}}
\usepackage{float}

\DeclareMathOperator{\Var}{Var}

\newcommand{\iidsim}{\stackrel {\textrm{ {\sc iid }}} {\sim} }

\newcommand*\diff{\mathop{}\!\mathrm{d}}
\renewcommand{\epsilon}{\varepsilon}
\renewcommand{\phi}{\varphi}

\newcommand{\cC}{\mathscr C}

\newcommand{\hH}{\mathscr H}

\newcommand{\fF}{\mathscr F}

\newcommand{\ZZ}{\mathsf Z}
\renewcommand{\SS}{\mathsf S}

\newcommand{\RR}{\mathbbm R}

\newcommand{\NN}{\mathbbm N}
\newcommand{\PP}{\mathbbm P}

\newcommand{\EE}{\mathbbm E}

\theoremstyle{plain}

\newtheorem{theorem}{Theorem}[section]

\newtheorem{lemma}{Lemma}[section]
\newtheorem{proposition}{Proposition}[section]

\theoremstyle{definition}

\newtheorem{example}{Example}[section]
\newtheorem{remark}{Remark}[section]

\newtheorem{assumption}{Assumption}[section]

\renewcommand{\underline}[1]{\text{\b{$#1$}}} 
  
\usepackage{caption}

\usepackage[flushleft]{threeparttable}

\usepackage{ dsfont }

\newcommand{\vertiii}[1]{{\left\vert\kern-0.25ex\left\vert\kern-0.25ex\left\vert #1 
    \right\vert\kern-0.25ex\right\vert\kern-0.25ex\right\vert}}
    
\usepackage{subcaption}
\usepackage{graphicx}  

\makeatletter
\newcommand*{\rom}[1]{\expandafter\@slowromancap\romannumeral #1@}
\makeatother

\usepackage{stackengine}

\begin{document}


\title{}

\date{\today}

\begin{center}
  \LARGE 
    A Theory of Saving under Risk Preference Dynamics
 
  \vspace{.6em}
  \normalsize
  Qingyin Ma\footnotesymbol{ISEM, Capital University of Economics and Business. \emph{Email:} \texttt{qingyin.ma@cueb.edu.cn}.}, \quad Xinxi Song\footnotesymbol{ISEM, Capital University of Economics and Business. \emph{Email}: \texttt{songxinxi@cueb.edu.cn}.},
  \quad Alexis Akira Toda\footnotesymbol{Department of Economics, Emory University. \emph{Email}: \texttt{alexis.akira.toda@emory.edu}.}
  \par 
  \vspace{.6em}
  \normalsize{\today}
\end{center}

\begin{abstract} 
	Empirical evidence shows that wealthy households have substantially higher saving rates and markedly lower marginal propensity to consume (MPC) than other groups. Existing theory cannot account for this pattern without jointly imposing restrictive assumptions on returns, discounting, and preferences. In this paper, we develop a general theory of optimal savings with preference shocks and identify a novel mechanism through which stochastic risk preferences reshape the asymptotic consumption and saving behavior. Specifically, the mere possibility of becoming less risk averse next period raises the value of carrying wealth forward, since future selves may be more willing to convert wealth into consumption. Unlike the classical precautionary saving motive, which typically arises from resource risks and weakens as wealth increases, this force remains operative even at arbitrarily high wealth levels, generating a persistent incentive to defer consumption and driving the asymptotic MPC to zero (i.e., a 100\% asymptotic saving rate). As a result, vanishing MPCs emerge as a generic implication of risk preference dynamics, rather than an artifact of restrictive assumptions, offering a theoretically robust and empirically consistent account of the persistently high saving rates and low MPCs observed among wealthy households.

    \vspace{.6em}
    
\end{abstract}


\section{Introduction}

In recent decades, a large body of empirical evidence has shown that wealthy 
households save a substantially larger fraction of their income and exhibit 
markedly lower marginal propensity to consume (MPC) than other groups. For 
instance, \cite{dynan2004rich} document a strong positive relationship between 
lifetime income and saving rates. Using Norwegian administrative data, 
\cite{fagereng2025saving} find that households in the top 1\% of the wealth 
distribution report saving rates far exceeding those of the median. Moreover, \citet{mian2025saving} document the emergence of a \emph{saving glut of the rich}, finding that since the 1980s, the top 1\% of wealthy households in US has sharply increased its saving rate from about 43\% to over 54\% of disposable income, while saving rates for the rest of the population have declined.

Because household saving behavior directly drives wealth accumulation, and 
wealthy households hold a disproportionate share of total 
assets,\footnote{\citet{saez2016wealth} 
	document that the top 0.1\% wealth share in the United States increased from 7\% in 1979 to 22\% in 2012. Moreover, \citet{piketty2019capital} estimate that the top 10\% wealth share in China rose from 40\% in 1995 to 67\% in 2015.} 
their saving patterns have first-order implications for 
the long-run distribution of wealth and the mechanisms that sustain inequality 
\citep{benhabib2018skewed}. From a macroeconomic perspective, even small 
variations in the MPC or risk-taking behavior of the wealthy can have outsized 
effects on aggregate saving rates and capital formation 
\citep{carroll2017distribution, lee2023wealth}, the equilibrium interest rates 
and asset prices \citep{mian2013household, mian2021indebted, 
patterson2023matching}, the transmission of fiscal and financial shocks 
\citep{cho2023umemployment, ampudia2024mpc, auclert2025fiscal}, 
and the design of macroprudential policy \citep{jappelli2014fiscal, 
crawley2023consumption}.

Despite these empirical regularities, existing studies remain limited 
in explaining the persistently high saving rates and low MPCs of wealthy 
households. The canonical incomplete market model predicts that, in a stationary equilibrium with constant risk aversion, discount factor, and return on wealth, the asymptotic saving rate is negative, contradicting the empirical evidence that wealthy households continue to accumulate assets indefinitely.\footnote{\label{fn:saving_rate}See, for example, 
	\cite{ma2021theory} and \cite{achdou2022income}. As is standard in the literature, we define the
	\textit{saving rate} as the change in net worth relative to total income, 
	excluding capital losses:
	\begin{equation*}
		s_{t+1} 
		= \frac{w_{t+1} - w_t}{\max \{(R_{t+1}-1)(w_t-c_t), 0\} + Y_{t+1}},
	\end{equation*}
	where $s_t$, $w_t$, $c_t$, $Y_t$, and $R_t$ denote the saving rate, wealth, 
	consumption, nonfinancial income, and the gross rate of return on wealth, 
	respectively. The \textit{asymptotic saving rate} refers to the  
	saving rate in the limit as $w_t \to \infty$.} 
Subsequent studies have introduced return risk and stochastic discounting to better reconcile theory with data \citep{benhabib2015wealth, ma2020income}, and recent work has further characterized the asymptotic saving behavior in such stochastic frameworks \citep{ma2021theory}. Yet, even in these richer dynamic settings, reproducing the observed low MPCs and persistently high saving rates requires restrictive assumptions that are jointly imposed on returns, discounting, and preferences. 
Consequently, the resulting explanation remains fragile and highly sensitive to the underlying parameter configuration, limiting its robustness.

A key limitation of the existing literature is that it abstracts from an important source of heterogeneity and uncertainty: the evolution of risk preferences themselves. Recent research highlights that households differ not only in income, wealth, and investment opportunities, but also in their \textit{attitudes toward risk}, and that these attitudes evolve over time. While the canonical savings model assumes a fixed risk aversion coefficient, a growing body of empirical and experimental evidence shows that risk attitude is heterogeneous across individuals and varies systematically over the life cycle and across business cycles \citep{anderson2008lost, cohen2007estimating, falk2018global, finkelstein2013what}. Moreover, risk tolerance comoves with individual financial circumstances, macroeconomic conditions, and asset market volatility \citep{guiso2018time, malmendier2011babies, sahm2012risk, schildberg2018risk}. Together, these findings suggest that agents continuously reassess their risk attitudes in response to both idiosyncratic and aggregate shocks.  

This evidence calls for a theory of intertemporal saving under dynamically evolving risk preferences. By abstracting from fluctuations in risk attitudes, existing theories overlook a key feedback between risk preferences and wealth accumulation, potentially distorting predictions of the consumption and saving behavior. Under constant risk preference, households save to insure against future resource scarcity, such as low income or poor investment returns. As wealth grows, however, this precautionary saving motive weakens, so  consumption can eventually exceed income, potentially driving saving rates negative at high wealth levels.

In this paper, we develop a general theory of optimal savings with preference shocks and identify a novel mechanism that reshapes asymptotic consumption and saving behavior. Whenever there is a positive probability that agents become less risk averse next period, retaining wealth today preserves resources for future states where wealth carries greater consumption value. Unlike the classical precautionary saving motive, which arises from future resource risks, this effect is driven by the evolution of risk preferences and remains operative even at arbitrarily high wealth levels. As a result, households have a persistent incentive to defer consumption, driving the asymptotic MPC to zero (i.e., a 100\% asymptotic saving rate). These dynamics stand in sharp contrast to existing theories, offering a parsimonious and robust explanation for the persistently high saving rates and low MPCs observed among wealthy households.

We begin by studying a generic optimal savings problem in which agents face 
preference shocks, and the rate of return, discount factor, and nonfinancial 
income are all allowed to be state-dependent, time-varying, and mutually 
dependent. We establish the existence and uniqueness of the optimal policy in a
suitable candidate space, under assumptions that have a transparent interpretation in terms of discounted average gross payoffs on assets and are straightforward to verify in practice.

We then characterize the properties of the optimal consumption policy. Under general specifications, we establish continuity and monotonicity, and analytically identify the wealth threshold above which saving occurs. To systematically study the behavior of wealthy households, we adopt a generalized CRRA framework, in which risk aversion is allowed to vary with the state. This specification aligns with the empirical evidence that individuals adjust their risk tolerance in response to changes in preferences, financial circumstances, and prevailing economic conditions, providing a more realistic foundation for studying intertemporal consumption and saving behavior. 

Incorporating such preference heterogeneity fundamentally changes the asymptotic dynamics of saving. In particular, we provide an analytical characterization of the asymptotic MPC and show that vanishing asymptotic MPCs arise under markedly weaker conditions than in the benchmark model without preference shocks. Whenever there is a positive probability of transitioning to a less risk averse state next period, agents retain a persistent incentive to defer consumption, transferring wealth toward these future states even at arbitrarily high levels of wealth. This effect drives the asymptotic MPC to zero and the asymptotic saving rate to 100\%. Importantly, this outcome is \emph{entirely independent of stochastic discounting, return risk, or income risk}, and thus represents a robust and general property of optimal savings behavior.

The asymptotic MPC is strictly positive at a given state \emph{only if} the probability of transitioning to a lower risk aversion state next period is zero. In such cases, risk attitudes evolve one-sidedly---either remaining constant or increasing---and positive asymptotic MPCs may arise, depending on the interaction between the stochastic discount factor and the return on wealth. If agents always transition to a higher risk aversion state from some state, then the asymptotic MPC at that state converges to one, and the optimal consumption function becomes \emph{nonconcave} in wealth, reflecting a strong incentive for consumption today in which the marginal value of transferring wealth to future selves vanishes.

We then complement our theoretical analysis with a quantitative investigation. First, we provide numerical illustrations that render the asymptotic results transparent and economically intuitive. Second, following \citet{koop1996impulse}, we compute state-dependent impulse response functions to uncover the key mechanisms through which preference shocks reshape consumption and saving dynamics.

Our quantitative analysis confirms that shocks to risk aversion generate economically significant and persistent responses in consumption and saving. Importantly, both the magnitude and persistence of these responses depend on households' current state and the persistence of preference shocks. In particular, the impacts vary systematically with wealth and prevailing risk attitudes, highlighting pronounced nonlinearities in the transmission of preference shocks across the wealth distribution.

Taken together, our framework provides a theoretically robust and empirically consistent account of the saving behavior of wealthy households, demonstrating how risk preference dynamics offer a natural mechanism for persistent wealth accumulation. While we do not explicitly model wealth inequality, our results point to a novel channel relevant for understanding wealth distribution in richer environments.

\subsection*{Related Literature}

The study of asymptotic saving behavior has its origins in the literature on the concavity of consumption functions. The classical work of \citet{carroll1996concavity} establishes that, under hyperbolic absolute risk aversion (HARA) preferences, the consumption function is concave in wealth, implying asymptotic linearity of consumption with respect to wealth. \citet{carroll2004theoretical} further emphasizes this asymptotic linearity and notes that zero asymptotic MPCs can arise in a standard optimal savings model when interest rates are negative. In a general incomplete market model with aggregate uncertainty, \citet{cao2020recursive} establishes concavity of the consumption function and shows that the MPC out of assets converges to a constant. In continuous time settings, \citet{achdou2022income} develop a comprehensive analysis of the Aiyagari-Bewley-Huggett model, establishing asymptotic linearity of consumption and showing that the saving rate of the wealthy is negative, while \citet{Roulleau-Pasdeloup2026} derives closed-form solutions for the consumption function in a classical savings problem, establishing concavity with respect to initial assets and permanent income and offering a sharp characterization of optimal consumption.

Recent studies have focused directly on the characterization of asymptotic MPCs. \citet{benhabib2015wealth} characterize the asymptotic MPCs analytically in an optimal savings model with capital income risk, where the return to wealth follows an {\sc iid} process. \citet{ma2021theory} build on and extend this framework by allowing returns and discount factors to be state-dependent, and jointly driven by a generic Markov process. Their contribution lies in delivering a systematic analysis of the asymptotic MPCs in this setting and identifying conditions under which optimal consumption is \textit{asymptotically linear} in wealth. In their framework, vanishing MPCs emerge only under a restrictive combination of conditions on returns, discounting, and preferences, limiting the empirical plausibility and robustness of the resulting explanation. 

Our paper contributes to this literature by developing a general theory of consumption and saving under stochastic risk preferences. As noted earlier, this induces vanishing asymptotic MPCs under conditions that are independent of stochastic discounting, return and income risks, and substantially weaker than those in the existing literature. Beyond these results, our framework also generates novel asymptotic outcomes when downward transitions in risk aversion are absent. In particular, the asymptotic MPC may converge to one and the consumption function may become nonconcave in wealth, departing sharply from the classical concavity-based intuition of \citet{carroll1996concavity}, \cite{benhabib2015wealth}, \cite{cao2020recursive}, \citet{ma2020income}, \cite{ma2021theory}, and \cite{Roulleau-Pasdeloup2026}.

Our paper is also related to the theoretical literature on optimal savings problems, a cornerstone of modern macroeconomics \citep{chamberlain2000optimal, kuhn2013recursive, li2014solving, benhabib2015wealth, cao2020recursive, ma2020income, Roulleau-Pasdeloup2026}. This literature has primarily focused on the existence, uniqueness, and computation of optimal policies under increasingly general environments. We contribute to this line of research by establishing a general optimal savings framework with preference shocks and deriving new results on the asymptotic behavior of consumption and saving. The asymptotic analysis developed here has no direct counterpart in the existing literature.

The rest of the paper is organized as follows. Section~\ref{s:or} formulates a 
generic optimal savings problem with preference shocks, and establishes 
existence, uniqueness, and key properties of the optimal policy. 
Section~\ref{s:asym} systematically studies its asymptotic properties and 
discusses the resulting implications for saving and wealth accumulation. Section~\ref{s:quant} presents a quantitative analysis that provides further intuition for our theory. Proofs are relegated to the Appendix.

\section{Optimality Results}
\label{s:or}

In this section, we establish fundamental optimality properties of a general  savings problem with preference shocks, encompassing a broad class of preference dynamics, and provide a rigorous foundation for the core asymptotic analysis of consumption and saving behavior in subsequent sections.

\subsection{Problem Statement}
\label{ss:ifp_problem}

Agents choose a consumption-wealth path $\{(c_t, w_t)\}_{t\geq 0}$ to solve the following optimal savings problem:
\begin{align}
	& \text{maximize}  
	&& \EE_0 \left\{ 
		\sum_{t = 0}^\infty 
		\left(\prod_{i=0}^t \beta_i \right) 
		u(c_t, Z_t)
	\right\}    \label{eq:value} \\
	&\text{subject to}   
	&& w_{t+1} = R_{t+1} (w_t - c_t) + Y_{t+1} 
	\quad \text{and} \quad 0 \leq c_t \leq w_t. \label{eq:trans_at}
\end{align}
Here, the initial condition $(w_0,Z_0) = (w,z)$ is given, $u$ denotes the 
utility function, $\{Z_t\}_{t \geq 0}$ is a time-homogeneous Markov 
chain taking values in a finite set $\ZZ$, $\{\beta_t\}_{t \geq 0}$ is the 
discount factor process with $\beta_0=1$, $\{R_t\}_{t \geq 1}$ is the gross 
rate of return on wealth, and $\{Y_t \}_{t \geq 1}$ is non-financial income. 
These stochastic processes satisfy
\begin{equation}
	\label{eq:law_exog}
	\beta_t = \beta \left(Z_{t-1}, Z_t, \epsilon_t \right), 
	\quad R_{t} = R \left(Z_{t-1}, Z_{t}, \epsilon_t \right)
	\quad \text{and} \quad
	Y_{t} = Y \left(Z_{t-1}, Z_{t}, \epsilon_t \right),
\end{equation}
where $\beta$, $R$, and $Y$ are nonnegative measurable functions. The 
innovations $\{\epsilon_t\}_{t\geq1}$ are independent and identically 
distributed. To properly capture important features of the real economy, we 
allow the supports of $\{Z_t\}$ and $\{\epsilon_t\}$ to be vector-valued, and 
we allow the support of $\{\epsilon_t\}$ to be continuous. 

In this setting, $\{Z_t\}$ captures both \emph{preference shocks} and \emph{broader idiosyncratic and aggregator factors} that jointly shape agents’ intertemporal tradeoffs and consumption and saving behavior. To simplify notation, we denote the partial derivative of utility with respect to consumption by
\begin{equation*}
	u'(c,z) := \partial u(c,z) \big/ \partial c.
\end{equation*}
In addition, we denote $\hat X$ as the next period value of a random variable $X$, and 
\begin{equation}
	\label{eq:nota1} 
	\EE_{w, z} 
	:= \EE \left[ \,\cdot \, \big| \, (w_0, Z_0) = (w, z) \right]
	\quad \text{and} \quad
	\EE_z 
	:= \EE \left[ \, \cdot \, \big| \, Z_0 = z \right].
\end{equation}

\subsection{Optimality: Definitions and Fundamental Properties}
\label{ss:opt}

For optimality, we assume $w_0>0$ and set the wealth space to 
$(0,\infty)$.\footnote{Assumptions~\ref{a:utility} and 
	\ref{a:spec_and_finite_exp}~\eqref{ac:fin_exp} introduced below
	imply $\PP_z \{ Y_t >0 \}=1$ for all $z\in \ZZ$ and $t \geq 1$. 
	Hence, $\PP_z \{ w_t > 0 \} = 1$ for all $z\in \ZZ$ and $t \geq 1$. 
	Therefore, excluding zero from the asset space has no effect on optimality.}
The state space for $\{(w_t, Z_t) \}_{t \geq 0}$ is then $\SS:= (0, \infty) 
\times \ZZ$. A \emph{feasible policy} is a Borel measurable function $c \colon 
\SS \to \RR$ with $0 < c(w,z) \leq w$ for all $(w,z) \in \SS$. A 
feasible policy $c$ and initial condition $(w,z) \in \SS$ generate a wealth
path $\{ w_t\}_{t \geq 0}$ via \eqref{eq:trans_at} when $c_t = c (w_t, Z_t)$ 
and $(w_0, Z_0) = (w,z)$. To proceed, we impose the following condition on the 
utility function.

\begin{assumption}
	\label{a:utility}
	For all $z \in \ZZ$, $u(\cdot,z)$ is twice differentiable on $(0,\infty)$ 
	with $u'(\cdot,z)>0$, $u''(\cdot,z)<0$, and $u'(c,z) \to \infty$ as $c\to 0$. 
\end{assumption}

Since we did not assume that the utility function is bounded or the discount 
factor $\beta_t$ is less than one, the lifetime utility \eqref{eq:value} may 
not be well defined in $\RR$. To overcome this issue, we define optimality by the overtaking criterion of \cite{brock1970axiomatic}. Given a feasible policy $c$, we consider the expected sum of utilities over a finite horizon up to time $H$, denoted by
\begin{equation*}
	V_{c,H} (w,z) := \EE_{w,z}\sum_{t=0}^{H} 
	\left(\prod_{i=0}^t \beta_{i} \right)
	u \left(c(w_{t},Z_{t}), Z_t \right).
\end{equation*}
For two feasible policies $c_1,c_2$, we say that $c_1$ \textit{overtakes} $c_2$ if 
\begin{equation*}
	\limsup_{H\to\infty} \left[
	    V_{c_2,H} (w,z) - V_{c_1,H} (w,z)
	\right] \leq 0
\end{equation*}
for all $(w,z)\in \SS$. We say that a feasible policy $c^*$ is 
\textit{optimal} if it overtakes any other feasible policy $c$. A feasible 
policy $c$ is said to satisfy the \textit{first order optimality condition} if
\begin{equation*}
	u'\left(c(w,z), z \right) 
	\geq \EE_z \hat \beta \hat R u'\big(c(\hat w,\hat Z), \hat Z\big) 
\end{equation*}
for all $(w,z) \in \SS$, and equality holds when $c(w,z) < w$. Here
\begin{equation*}
	\hat w = \hat R(w-c(w,z)) + \hat Y.
\end{equation*}
Noting that $u'(\cdot, z)$ is decreasing, the first order condition can be compactly stated as
\begin{equation}\label{eq:foc}
	u'\left(c(w,z), z\right) = \max \left\{
	    \EE_z \hat \beta \hat R u'\big(c(\hat w,\hat Z), \hat Z\big), 
	    u'(w, z)
	\right\}
\end{equation}
for all $(w,z) \in \SS$. A feasible policy $c$ is said to satisfy the 
\textit{transversality condition} if, for all $(w,z) \in \SS$, 
\begin{equation}\label{eq:tvc}
	\lim_{t \to \infty} \EE_{w,z} 
	    \left(\prod_{i=0}^t \beta_{i}\right)
	u'\left(c(w_t, Z_t), Z_t\right) w_t=0.
\end{equation}
The following result demonstrates that the first order and transiversality 
conditions are sufficient for optimality. The proof is similar to 
Proposition~15.2 and Theorem~15.3 of \cite{toda2025essential} and thus omitted.

\begin{proposition}[Sufficiency of first order and transversality conditions]
	\label{pr:suff_euler_tvc}
	If Assumption~\ref{a:utility} holds, then every feasible policy satisfying 
	the first order and transversality conditions is an optimal policy.
\end{proposition}

\subsection{Existence and Computability of Optimal Consumption}
\label{ss:exis_uniq}

We now study the existence, uniqueness, and computability of a feasible policy 
satisfying the first order condition~\eqref{eq:foc}. We set aside the 
discussion of the transversality condition \eqref{eq:tvc}, as it requires 
only mild assumptions not essential for the core analysis, and its verification follows essentially the same steps as Proposition~2.2 of \cite{ma2020income}.


Let $\cC$ be the set of continuous functions $c: \SS \to \RR_+$ such that 
$w \mapsto c(w,z)$ is increasing for all $z\in \ZZ$, $0< c(w,z) \leq w$ for all 
$(w,z)\in \SS$, and 
\begin{equation}\label{eq:upc_bd}
	\sup_{(w,z)\in \SS} \left|
	u^{\prime}\big(c(w,z), z\big) - u^{\prime}\big( w, z \big)
	\right| < \infty.
\end{equation}
We pair $\cC$ with a distance as follows: For each $c_1, c_2 \in \cC$, define
\begin{equation*}
	\rho(c_1,c_2) := \sup_{(w,z)\in \SS} \left|
	    u'\big(c_1(w,z), z\big) - u'\big(c_2(w,z), z\big)
	\right|,
\end{equation*}
which evaluates the maximal difference in terms of marginal utility. While
elements of $\cC$ are not generally bounded, $\rho$ is a
valid metric on $\cC$. In particular, $\rho$ is finite on $\cC$ since the 
triangular inequality implies that
$\rho(c_1,c_2) \leq \left\| u' \circ c_1 - u' \right\| 
+ \left\| u' \circ c_2 - u' \right\|$, where $\|\cdot\|$ is the standard 
supremum norm, and the last two terms are finite by~\eqref{eq:upc_bd}.
In Appendix~\ref{s:proof_or}, we show that $(\cC, \rho)$ is a complete metric 
space. 

We aim to characterize the optimal policy as the fixed point of the 
\emph{time iteration operator} $T$ defined as follows: for fixed $c \in \cC$ and 
$(w,z) \in \SS$, the value of the image $Tc$ at $(w,z)$ is defined as the 
$\xi \in (0,w]$ that solves
\begin{equation}
	\label{eq:tio}
	u'\left(\xi, z\right) = \max \left\{
		\EE_z \hat \beta \hat R u'\left( 
		    c\big(\hat R (w-\xi) + \hat Y, \hat Z\big), \hat Z
		\right),
		u^{\prime}\big(w, z\big)
	\right\}.
\end{equation}
To show that $T$ is a well defined self-map on $\cC$ and characterize the optimal policy via this operator, we impose several key assumptions. Let $P(z, \hat z)$ denote the one-step transition probability from $z$ to $\hat z$. For 
$\theta \in \RR$, we define the matrix $K(\theta)$ as follows. For each 
$z, \hat z \in \ZZ$, let
\begin{equation}\label{eq:Kmat}
	K_{z\hat z}(\theta) 
	:= P(z,\hat z) \int \beta(z, \hat z, \hat \epsilon) 
	R(z, \hat z, \hat \epsilon)^\theta \pi (\diff \hat \epsilon),
\end{equation}
where $\pi$ is the probability distribution of $\{\epsilon_t\}$.\footnote{The 
	matrix $K(\theta)$ is expressed as a function on $\ZZ \times \ZZ$ in 
	\eqref{eq:Kmat} but can be represented in traditional matrix notation by 
	enumerating $\ZZ$. }
For a square matrix $A$, we use $r(A)$ to denote its spectral radius, defined as
\begin{equation*}
	r(A) := \max \{|\lambda|: \lambda \text{ is an eigenvalue of } A\}.
\end{equation*}
In other words, $r(A)$ is the largest absolute value of all its eigenvalues. 

\begin{assumption}
	\label{a:spec_and_finite_exp}
	The following conditions hold:
	\begin{enumerate}
		\item\label{ac:fin_exp} 
		For all $z\in \ZZ$, we have $\EE_z u'(\hat Y,\hat Z) < \infty$ and 
		$\EE_z \hat \beta \hat R u'(\hat Y,\hat Z) < \infty$.
		
		\item\label{ac:spec} 
		$r(K(1)) < 1$, where the matrix $K(\theta)$ 
		is defined by~\eqref{eq:Kmat}.
	\end{enumerate}
\end{assumption}

Following \cite{ma2020income}, $\ln r(K(1))$ can be interpreted as the the 
asymptotic growth rate of average discounted gross payoff on assets, expressed 
in present-value terms. Accordingly, 
Assumption~\ref{a:spec_and_finite_exp}~\eqref{ac:spec} requires this  
rate to be negative to ensure that wealth does not diverge asymptotically. This 
does not preclude the possibility that $\beta_t R_t \geq 1$ at any given $t$. 
Rather, it substantially generalizes the standard assumption $\beta R < 1$ in 
the classical optimal savings problem, where both $R_t \equiv R$ and 
$\beta_t\equiv \beta$ are constant.

The following theorem establishes the existence and uniqueness of a candidate 
policy that satisfies the first order optimality condition.

\begin{theorem}[Existence, uniqueness, and computability of optimal policies]
	\label{t:opt}
	If Assumptions~\ref{a:utility}--\ref{a:spec_and_finite_exp} hold, then 
	the following statements are true:
	\begin{enumerate}
		\item $T: \cC \to \cC$ is well defined and has a unique fixed point $c^{*} \in \cC$.
		\item\label{tr:conv} For all $c \in \cC$, we have 
		$\rho(T^{k}c,c^{*}) \to 0$ as $k \to \infty$.
	\end{enumerate}
\end{theorem}

Part~\eqref{tr:conv} establishes that, under our assumptions, the familiar time iteration algorithm is globally convergent for any initial policy within the candidate class $\cC$. As we have deferred the discussion of the transversality condition, which is orthogonal to the central focus of this paper and can be readily incorporated when necessary, we will henceforth refer to $c^*$ in Theorem~\ref{t:opt} as the \textit{optimal consumption function}.

\subsection{Fundamental Properties of Consumption and Saving}
\label{ss:ifp_properties}

We next examine the properties of the optimal consumption function 
characterized in Theorem~\ref{t:opt}. Throughout, 
Assumptions~\ref{a:utility}--\ref{a:spec_and_finite_exp} are maintained. The 
following propositions establish the monotonicity of the consumption and saving
functions.

\begin{proposition}[Monotonicity with respect to wealth]
	\label{pr:monotonea}
	The optimal consumption and savings functions $c^*(w,z)$ and 
	$i^*(w,z) := w - c^*(w,z)$ are increasing in $w$.
\end{proposition}

\begin{proposition}[Monotonicity with respect to income]
	\label{pr:monotoneY} 
	If $\{ Y_{1t} \}$ and $\{ Y_{2t} \}$ are two income processes satisfying 
	$Y_{1t}\leq Y_{2t}$ for all $t$ and $c_1^*$ and $c_2^*$ are the
	corresponding optimal consumption functions, then $c_1^* \leq c_2^*$
	pointwise on $\SS$.
\end{proposition}

The next proposition demonstrates that the borrowing constraint is binding if 
and only if wealth is below a certain threshold.

\begin{proposition}[Threshold for saving decision]
	\label{pr:binding}
	For all $c \in \cC$, there exists a threshold $\bar{w}_c(z)$ such that 
	$Tc(w,z) = w$ if and only if $w \leq \bar{w}_c (z)$. In particular, letting 
	\begin{equation*}
		\bar w(z) := u'(\cdot, z)^{-1} \left[
		    \EE_z \hat \beta \hat R u'\big(c^*(\hat Y,\hat Z), \hat Z\big)
		\right],
	\end{equation*}
	we have $c^*(w,z) = w$ if and only if $w \leq \bar{w}(z)$.
\end{proposition}

\section{Asymptotic Properties of Consumption and Saving}
\label{s:asym}

In this section, we develop the core analysis of this paper. We systematically study the asymptotic properties of consumption and saving in the presence of risk aversion shocks. Throughout, we maintain Assumptions~\ref{a:utility}--\ref{a:spec_and_finite_exp}. To build intuition before presenting the general results, we start with a simple illustrative example.

\subsection{A Two-Period Heuristic Example}\label{ss:heuristic}
Consider a two-period model with two risk aversion states, $\gamma_{t} \in \{\gamma_L, \gamma_H\}$, where $\gamma_L < \gamma_H$ and $\gamma_L,\gamma_H \neq 1$. Suppose that agents begin in state $\gamma_H$ in period 0 and will transition to state $\gamma_L$ in period 1. They choose $(c_0, c_1)$ to solve:
\begin{align*}
	& \text{maximize}
	&& \frac{c_{0}^{1-\gamma_{H}}}{1-\gamma_{H}} +\beta  \frac{c_{1}^{1-\gamma_{L}}}{1-\gamma_{L}} \\
	&\text{subject to}
	&& 0 \leq c_{0} \leq w_{0} \quad \text{and} \quad  
	c_{1}=R(w_{0}-c_{0}),
\end{align*}
where $\beta>0$ is the discount factor and $R>0$ is the gross return on savings, both constant. Letting $(w_0, c_0) = (w,c)$, the first order optimality condition in period 0 is
\begin{equation*}
	c^{-\gamma_{H}}=\beta  R^{1-\gamma_{L}} \left(w-c\right)^{-\gamma_{L}}.
\end{equation*}
Since $w - c \le w$, we obtain $c^{-\gamma_{H}} \geq  \beta R^{1-\gamma_{L}} w^{-\gamma_{L}}$. Dividing both sides by $w^{-\gamma_{H}}$ yields a bound on the period-0 optimal consumption function, denoted by $c^*(w,\gamma_H)$:
\begin{equation*}
	0 \leq \frac{c^*(w,\gamma_H)}{w}\leq \left(\beta R^{1-\gamma_{L}}\right)^{-1/\gamma_{H}}w^{\gamma_{L}/\gamma_{H}-1}.
\end{equation*}
Letting $w \to \infty$ and using $\gamma_L < \gamma_H$, it follows that
\begin{equation*}
	\lim_{w \to \infty} \frac{c^*(w,\gamma_H)}{w} = 0.
\end{equation*}
Hence, the marginal propensity to consume (MPC) vanishes asymptotically, implying that infinitely wealthy agents save 100\% of their total income in period 0. Intuitively, the expectation of transitioning to a less risk averse regime induces agents to accumulate wealth at an increasing rate to rebalance their consumption across different risk preference states. This suggests that the saving behavior is driven not only by return and income risks, but also by the evolution of risk preferences themselves.

\subsection{Setups}\label{ss:formu_asym}
The heuristic example above illustrated how risk aversion shocks can fundamentally influence the saving behavior of wealthy agents. We now formalize this intuition in a general framework. To that end, we depart from the standard CRRA specification by allowing risk aversion to vary with the state and specify the utility function as 
\begin{equation}
	\label{eq:CRRA}
	u(c, z) =
	\begin{cases}
		\frac{c^{1-\gamma(z)}}{1 - \gamma(z)}, &
		\text{if } \gamma(z) > 0 \text{ and } \gamma(z) \neq 1, \\
		\log c, & \text{if } \gamma(z) = 1,
	\end{cases}
\end{equation}
where $\gamma(z)$ is a state-dependent measure of relative risk aversion. This formulation captures the idea that agents may become more or less tolerant of risk depending on idiosyncratic or aggregate shocks.

Assume that the Markov process $\{Z_t\}$ can be decomposed as $Z_t = (\bar Z_t, \tilde Z_t)$, where $\{\bar Z_t\}$ and $\{\tilde Z_t\}$ are independent Markov processes taking values in $\bar{\ZZ} := \{\bar z_1,\dots,\bar z_N\}$ and $\tilde{\ZZ} := \{\tilde z_1,\dots,\tilde z_M\}$, with transition matrices $\bar P = (\bar p_{ii'})_{1\le i,i' \le N}$ and $\tilde P = (\tilde p_{jj'})_{1\le j,j' \le M}$, respectively. In this setting, the transition probability of $\{Z_t\}$ satisfies $P = \bar P \otimes \tilde P$, where $\otimes$ denotes the Kronecker product. In addition, we assume that risk aversion is driven solely by $\{\bar Z_t\}$ and satisfies
\begin{equation*}
	0<\gamma(\bar{z}_{1})< \dots<\gamma(\bar{z}_{N}).
\end{equation*}
%
%
%
With slight abuse of notation, we denote $\gamma_i :=\gamma(\bar z_i)$ for 
$i=1, \dots, N$. For fixed $i$, we define the $M\times M$ matrix $Q_i$ with 
entries
\begin{equation*}
	Q_i(j,k)
	= \EE_{z_{ij}, z_{ik}} \beta(z_{ij}, z_{ik}, \hat \epsilon) 
	R(z_{ij}, z_{ik}, \hat \epsilon)^{1-\gamma_i},
\end{equation*}
where $z_{ij} = (\bar z_i, \tilde z_j)$, and the expectation is taken with 
respect to $\hat \epsilon$ conditional on $(Z,\hat Z) = (z_{ij}, z_{ik})$. 
That is, the elements of $Q_i$ represent the corresponding expected future values, 
conditional on current and future realizations of the exogenous state. 
We then set
\begin{equation}\label{eq:Gi}
	G_i := \bar p_{ii} (\tilde P \circ Q_i),
\end{equation}
%
where $\tilde P\circ Q_i$ represents the Hadamard product (i.e., entry-wise 
multiplication) of $\tilde P$ and $Q_i$. Hence $G_i$ is 
obtained by multiplying each entry of $\tilde P \circ Q_i$ by $\bar p_{ii}$.

Recall that a square matrix $A$ is called \textit{reducible} if there is a permutation matrix $B$ such that $B'AB$ is block upper triangular with at least two blocks:
\begin{equation*}
	B'AB = 
	\begin{bmatrix}
		A_{11} & A_{12} \\
		0 & A_{22}
	\end{bmatrix}
\end{equation*}
where $A_{11}$ and $A_{22}$ are square matrices of size at least one. We call a square matrix \textit{irreducible} if it is not reducible. 
Recall that $r(A)$ denotes the spectral radius of $A$.

\subsection{Downward Risk Aversion Transition and Vanishing Asymptotic MPCs}\label{ss:ampc=0}
Throughout, we adopt the convention that $\sum_{\ell=1}^{i-1} \bar p_{i\ell}=0$ for $i=1$. We now present the central results of this paper, establishing that in the presence of stochastic risk preferences, the asymptotic MPC vanishes as a generic consequence of optimal savings behavior.

\begin{theorem}[Vanishing asymptotic MPCs]
	\label{t:0ampc_thm}
	Fix $i \in \{1,\dots, N\}$. Suppose either\footnote{If 
		$\{\beta_t\}$ and $\{R_t\}$ are 
		strictly positive stochastic processes, as is typically assumed in 
		applications, then $\PP_{(z,\hat z)} (\hat\beta \hat R > 0) = 1$ for all 
		$z,\hat z\in\ZZ$. In particular, by definition, 
		\begin{equation*}
			\PP_{(z,\hat z)} (\hat\beta \hat R > 0) = 
			\PP \left(
			\beta(Z,\hat Z,\hat\epsilon) R(Z,\hat Z,\hat\epsilon) > 0 
			\Big| (Z,\hat Z)=(z,\hat z)
			\right).
	    \end{equation*}}
	\begin{enumerate}
	    \item\label{c1} $\sum_{\ell=1}^{i-1}\bar{p}_{i\ell}>0$ and 
	    $\PP_{(z,\hat z)}(\hat \beta \hat R > 0) > 0$ for all 
	    $z,\hat z\in \ZZ$, or
	    
	    \item\label{c2} $\bar{p}_{ii}>0$, $G_i$ is irreducible, and $r(G_i) \geq 1$.
	\end{enumerate}
	Then we have
	\begin{equation}
		\label{eq:0ampc_thm}
		\lim_{w\to\infty} \frac{c^*(w,z_{ij})}{w} =0
		\quad \text{for all }\, j=1,\dots, M.
	\end{equation}
\end{theorem}

Here, $\sum_{\ell=1}^{i-1} \bar p_{i\ell}$ is the probability of transitioning from $\gamma_i$ to a lower risk aversion state in the next period. Theorem~\ref{t:0ampc_thm} shows that whenever there is a positive probability of transitioning to a lower risk aversion state, however small, and $\beta R > 0$ occurs with positive probability, the asymptotic MPC is necessarily zero. 

Theorem~\ref{t:0ampc_thm} formalizes the intuition developed in the heuristic example of Section~\ref{ss:heuristic}. The mere prospect of becoming less risk averse next period raises the value of carrying wealth forward. As a result, agents retain wealth not only to hedge against stochastic discounting, return and income risks, but also to preserve resources that might be particularly valuable under future risk preference states. Because this effect remains relevant even at very high wealth levels, consumption grows more slowly than wealth, breaking the homothetic scaling property and making vanishing asymptotic MPCs a necessary outcome of optimal intertemporal choice.

\begin{remark}
	Theorem~\ref{t:0ampc_thm} uncovers a fundamentally new mechanism for vanishing MPCs. Existing literature typically relies on restrictive parametric or spectral radius conditions that jointly constrain returns, discount factors, and preferences to force vanishing MPCs \citep[see e.g.,][]{carroll2004theoretical, ma2021theory}. By contrast, we show that with stochastic risk preferences, the asymptotic MPC is zero whenever agents internalize the possibility of a one-period risk aversion decay. Consequently, vanishing MPCs emerge as a fundamental implication of optimal savings behavior under risk preference dynamics, rather than as an artifact of restrictive conditions.   
\end{remark}

\begin{remark}
	It is also worth emphasizing that Theorem~\ref{t:0ampc_thm} is stated as an \emph{either-or} result. Condition~\eqref{c2} can be viewed as a \textit{generalized knife-edge condition} involving returns, discounting, and risk aversion dynamics. Relative to condition~\eqref{c1}, it imposes substantially stronger restrictions and is therefore considerably less likely to hold in applications. Nevertheless, condition~\eqref{c2} is of independent theoretical interest. When $\bar p_{ii}=1$, it reduces to a preference shock analogue of Theorem~3 in \cite{ma2021theory}. Hence, even when downward transitions in risk aversion are ruled out, Theorem~\ref{t:0ampc_thm} yields new results and insights for economies with stochastic risk preferences, a class of environments that lies beyond the scope of existing theory.
\end{remark}

Given that transition matrices in empirical applications often feature strictly positive entries, Condition~\eqref{c1} is frequently satisfied, making our theory broadly relevant and motivating the following result.

\begin{theorem}[Vanishing asymptotic MPCs under positive transition probability]
	\label{t:0ampc_thm_pos}
	If every entry of $\bar{P}$ is positive and 
	$\PP_{(z,\hat z)}(\hat \beta \hat R>0) >0$ for all 
	$z,\hat z\in \ZZ$, then
	\begin{equation}
		\label{eq:0ampc_prop}
		\lim_{w\to\infty} \frac{c^*(w,z_{ij})}{w} = 0
		\quad \text{for all } z_{ij} \in \ZZ \text{ with } i \neq 1.
	\end{equation}
	If in addition $G_1$ is irreducible and $r(G_1) \geq 1$, then 
	\eqref{eq:0ampc_prop} holds for all $z_{ij} \in \ZZ$.
\end{theorem}

Taken together, Theorems~\ref{t:0ampc_thm}--\ref{t:0ampc_thm_pos} establish a sharp theoretical boundary: unless in highly exceptional cases, the asymptotic MPC can be  positive \emph{only if} the probability of becoming less risk averse next period is zero. For the asymptotic MPCs to remain positive across all states, the probability transition matrix of risk aversion must be upper triangular, an exceptional 
case that rarely arises in practice, as it would require agents to only become 
weakly more risk averse over time. These observations are formalized in the 
next result. The proof follows immediately from Theorem~\ref{t:0ampc_thm}. 

\begin{proposition}[Impossibility of positive asymptotic MPCs without downward transitions]
	\label{cor:special_trans}
	If $\PP_{(z,\hat z)}(\hat \beta \hat R>0) >0$ for all 
	$z,\hat z\in \ZZ$, and there exists $z_{ij} \in \ZZ$ such that 
	\begin{equation}\label{eq:pos_ampc}
		\lim_{w \to \infty} \frac{c^*(w,z_{ij})}{w} > 0,
	\end{equation}
	then $\bar p_{i1} = \cdots = \bar p_{i,i-1} = 0$. In particular, if for all 
	$i>1$, there exists $j$ and $z_{ij}\in \ZZ$ such that \eqref{eq:pos_ampc} 
	holds, then $\bar P$ is upper triangular.
\end{proposition}

\subsection{Consumption Behavior without Downward Risk Aversion Transitions}
\label{ss:ampc>0}
We now turn to the complementary class of economies in which downward transitions in risk aversion are ruled out (i.e., $\sum_{\ell=1}^{i-1}\bar p_{i\ell} = 0$). Although such cases are less common in applications, they are important for completing the characterization of asymptotic consumption behavior. As we show below, the resulting outcomes depend critically on the persistence of the current risk aversion state $\bar p_{ii}$, generating qualitatively distinct regimes corresponding to $\bar p_{ii}=0$, $\bar p_{ii}=1$, and $\bar p_{ii}\in(0,1)$, which we analyze in turn.

\begin{theorem}[Unit asymptotic MPCs under upward risk aversion transitions]
	\label{t:non0_ampc2}
	Given $i\in \{1, \dots, N\}$, if $\sum_{\ell=1}^{i}\bar{p}_{i\ell}=0$ 
	and there exists $m>0$ such that $R \geq m$ with probability one, then 
	\begin{equation*}
		\lim_{w \to \infty} \frac{c^{*}(w,z_{ij})}{w} = 1
		\quad \text{for all }\, j=1,\dots, M.
	\end{equation*}
\end{theorem}

Theorem~\ref{t:non0_ampc2} reveals a striking behavioral shift.  When agents expect their risk aversion to rise over time (i.e., $\sum_{\ell=1}^{i}\bar{p}_{i\ell}=0$), the marginal value of transferring additional wealth into the future becomes negligible at high wealth levels. Intuitively, future selves are expected to be increasingly cautious and therefore have stronger incentives to preserve resources on their own. As a result, the current self has little incentive to accumulate additional wealth, causing the asymptotic MPC to converge to one. This outcome stands in sharp contrast both to the vanishing MPCs generated by downward risk aversion transitions and to the asymptotically linear consumption behavior emphasized in the existing literature. Moreover, it has profound implications for the structure of the optimal policy, as established below.

\begin{proposition}[Nonconcavity of the consumption policy]
	\label{pr:non-concave}
	Given $i\in \{1, \dots, N\}$, if the assumptions of 
	Theorem~\ref{t:non0_ampc2} hold, then $c^*(w,z_{ij})$ is not concave in $w$ 
	for all $j \in \{1, \dots, M\}$.
\end{proposition} 

\begin{remark}
	A central insight of the classical savings literature is that consumption policies are typically concave in wealth under HARA preferences \citep{carroll1996concavity}. This property continues to hold even in environments with capital income risk and stochastic discounting \citep{ma2020income}. Proposition~\ref{pr:non-concave} shows that introducing stochastic risk preferences can fundamentally overturn this conclusion. When risk aversion is sure to increase next period, the consumption policy becomes nonconcave, indicating that the classical connection between precautionary saving and concave consumption policies breaks down once risk attitudes evolve over time. 
\end{remark}

Now consider the case $\bar p_{ii} = 1$. For each $i\in\{1,\dots, N\}$ and $x \in [1,\infty)^M$, define
\begin{equation}
	\label{eq:F_oper0}
	(F_ix)(\tilde z_{j}):= \left(
	    1+(G_ix)(\tilde z_{j})^{1/\gamma_{i}}
	\right)^{\gamma_{i}}, 
	\qquad j \in \{1,\dots, M\}.
\end{equation}
The next result characterizes the asymptotic consumption behavior in this polar case.

\begin{theorem}[Positive asymptotic MPCs under strictly persistent risk aversion]
	\label{t:non0_ampc1}
	Given $i\in \{1, \dots, N\}$,
	if $\bar{p}_{ii}=1$ and $r(G_{i}) < 1$, then $F_i$ in \eqref{eq:F_oper0} 
	has a unique fixed point $x_i^{*}=(x_i^{*}(\tilde z_{j}))_{j=1}^{M}$ in 
	$[1,\infty)^{M}$ and
		\begin{equation}\label{eq:non0_ampc1}
			\lim_{w \to \infty} \frac{c^{*}(w,z_{ij})}{w} 
			= x_i^*(\tilde z_{j})^{-1/\gamma_{i}} > 0
			\quad \text{for all }\, j=1,\dots, M.
		\end{equation}
\end{theorem}

When $\bar p_{ii}=1$, the possibility of future changes in risk attitudes disappears along the relevant histories. In this case, the mechanism generated by risk preference dynamics is inactive, and the asymptotic MPC is determined by the interaction between returns, discounting, and the prevailing risk aversion level. The result yields a strictly positive asymptotic MPC characterized by the fixed point problem above.

Having characterized the two polar regimes $\bar p_{ii}=0$ and $\bar p_{ii}=1$, we now turn to the intermediate case $\bar p_{ii}\in(0,1)$, where the current risk aversion state is persistent but not absorbing. 

\begin{theorem}[Positive asymptotic MPCs under imperfect persistence]
	\label{t:non0_ampc3}
	Given $i\in \{1, \dots, N\}$, if (1) 
	$\sum_{\ell=1}^{i-1}\bar{p}_{i\ell}=0$, (2) $\bar{p}_{ii}\in (0,1)$ and 
	$r(G_i) < 1$, and (3) there exists $m, b >0$ such that $R \geq m$ and $Y\geq b$ 
	with probability one, then \eqref{eq:non0_ampc1} holds.
\end{theorem}

Establishing Theorem~\ref{t:non0_ampc3} is substantially more involved than the polar cases above. To construct upper and lower asymptotic MPC bounds and show that they coincide, we develop an important parametric monotonicity result for optimal consumption with respect to risk attitudes (Proposition~\ref{pr:pm}), which is of independent interest for the analysis of dynamic optimization problems with stochastic risk preferences. 

Taken together, the results in this section show that stochastic risk preferences do far more than generate vanishing asymptotic MPCs. Once risk attitudes evolve over time, asymptotic consumption behavior can exhibit qualitatively different regimes, ranging from vanishing MPCs to unit MPCs and strictly positive interior MPCs. Therefore, risk preference dynamics substantially enrich the set of possible long-run consumption outcomes and give rise to forms of asymptotic behavior that are absent from existing theories based on fixed risk attitudes. 
Table~\ref{tab:ampcs} summarizes the core results of Section~\ref{s:asym}.

\begin{landscape}
	\centering
	\vspace*{\fill}
	\begin{table}[htbp]
		\caption{Summary of Asymptotic MPCs  under Risk Preference Dynamics}
		\vspace{-.3cm}
		\label{tab:ampcs}
		\begin{minipage}{\linewidth}
			\centering
			\resizebox{1\textwidth}{!}{%
				\begin{tabular}{p{3.8cm}p{3.5cm}p{1cm}p{3.5cm}|p{3.5cm}p{3.5cm}p{3.6cm}}
					\specialrule{1.5pt}{0pt}{0pt}
					\multicolumn{4}{c|}{} & 
					\makecell[c]{\rule{0pt}{18pt}} &
					\makecell[c]{\rule{0pt}{18pt}$\sum_{\ell=1}^{i-1}\bar{p}_{i\ell} =0$} & 
					\makecell[c]{\rule{0pt}{18pt}} \\
					\noalign{\kern 6pt}\cline{5-7}\noalign{\kern -16pt}
					\makecell{\textbf{Key} \\ \textbf{Assumptions} \vspace{0.5cm}
					} & 
					\makecell{$\sum_{\ell=1}^{i-1}\bar{p}_{i\ell} >0$ 
						\vspace{0.5cm}} & 
					\makecell{or \vspace{0.5cm} } & 
					\makecell{$\bar{p}_{ii}>0$; \\ $r(G_{i})>1$ \vspace{0.5cm}} & 
					\makecell[c]{\rule{0pt}{27pt}$\bar{p}_{ii}=0$} & 
					\makecell[c]{\rule{0pt}{27pt}$\bar{p}_{ii}=1$} & 
					\makecell[c]{\rule{0pt}{27pt}$0<\bar{p}_{ii}<1$} \\
					\hline
					\makecell{\textbf{Interpretation} } & 
					\makecell{Less risk averse \\ with positive \\ probability} & 
					\makecell{} & 
					\makecell{Generalized \\ knife-edge \\ condition} & 
					\makecell{Strictly more \\ risk averse \\ with probability 1} & 
					\makecell{Same risk aversion \\ with probability 1} & 
					\makecell{Weakly more \\ risk averse with \\ positive probability} \\
					\hline
					\makecell{\textbf{Other}  \\ \textbf{Assumptions}} & 
					\makecell{$\PP_{(z,\hat z)}(\hat \beta \hat R >0)>0$ \\
						for all $z,\hat z \in \ZZ$} & 
					\makecell{} & 
					\makecell{$G_i$ is irreducible} & 
					\makecell{$R_t \geq m$ \\ for some $m>0$} & 
					\makecell{$r(G_{i})<1$} & 
					\makecell{$r(G_{i})<1$; \\ $R_t \geq m$, $Y_t \geq b$ \\ 
						for some $m,b > 0$} \\
					\specialrule{1.5pt}{0pt}{0pt}
					\makecell{\textbf{Asymptotic} \\ \textbf{MPCs}} & 
					\makecell{0} & 
					\makecell{} & 
					\makecell{0} & 
					\makecell{1} & 
					\makecell{$(0,1]$} & 
					\makecell{$(0,1]$} \\
					\hline
					\makecell{\textbf{References}} & 
					\makecell{Theorem 3.1} & 
					\makecell{} & 
					\makecell{Theorem 3.1} & 
					\makecell{Theorem 3.3} & 
					\makecell{Theorem 3.4} & 
					\makecell{Theorem 3.5} \\
					\specialrule{1.5pt}{0pt}{0pt}
				\end{tabular}%
			}
			\parbox{1\textwidth}{\vspace{.1cm}\small{Note: Results are reported conditional on the current risk aversion state $\gamma_i$ and hold for each $\gamma_i$. Here, $\bar p_{ij}$ denotes the one-period transition probability from $\gamma_i$ to $\gamma_j$. The matrix $G_i$ is defined in \eqref{eq:Gi}} and $r(G_i)$ is its spectral radius.}
		\end{minipage}
	\end{table}
	\vspace*{\fill}
\end{landscape}

\subsection{Examples}\label{ss:ex}
In what follows, we illustrate how stochastic risk preferences alter the asymptotic dynamics of several benchmark savings models. In each example, the introduction of risk preference shocks generates outcomes that differ sharply from those obtained under constant risk aversion. For each $z\in \ZZ$, we denote
\begin{equation*}
	\bar c(z) := \lim_{w\to\infty} \frac{c^*(w,z)}{w}.
\end{equation*}

\begin{example}[Downward risk aversion transitions]
	Suppose $\beta R > 0$ with positive probability at each 
	state and $\bar p_{i,i-1}>0$ 
	for all $i>2$, implying that risk aversion follows a persistent Markov 
	chain with strictly positive probability of downward transitions between 
	adjacent states. Then, by Theorem~\ref{t:0ampc_thm}, we have 
	$\bar c(z_{ij})=0$ for all $i>1$ and all $j$. For $i=1$, if $G_1$ is 
	irreducible with $r(G_1)\geq 1$, then $\bar c(z_{1j})=0$ for all $j$, and 
	if instead $r(G_1)<1$, then $\bar c(z_{1j})\in (0,1]$ for all $j$. Hence, vanishing asymptotic MPCs arise throughout the economy except possibly at the lowest risk aversion state.
\end{example}

\begin{example}[Strictly upward transitions]
	Suppose that, at some state, agents transition with probability one to a more risk averse state in the next period, and $R_t$ has a positive lower bound (which can be arbitrarily close to zero). Then Theorem~\ref{t:non0_ampc2} implies that consumption asymptotically tracks wealth one-for-one at that state, so the asymptotic MPC equals one. Moreover, Proposition~\ref{pr:non-concave} shows that the optimal consumption policy cannot be concave in wealth.
\end{example}

\begin{example}[Constant returns without capital income risk]
	Suppose the gross rate of return $R_t$ is constant at $R$. By Proposition~8 of \citet*{ma2021theory}, in the absence of preference shocks, a zero asymptotic MPC can occur only when $R<1$. This condition is difficult to reconcile with the observed asset returns, as it requires a negative risk-free interest rate. By contrast, once stochastic risk preferences are introduced, Theorem~\ref{t:0ampc_thm} implies that vanishing asymptotic MPCs arise whenever there is a positive probability of transitioning to a less risk averse state and $\beta R>0$ occurs with positive probability, without any further restriction on $R$. Therefore, the mechanism generating vanishing MPCs no longer relies on economically implausible return assumptions.
\end{example}

\begin{example}[Classical optimal savings problem with preference shocks]
	\label{ex:clas_ops}
	Consider the classical optimal savings problem, where $R_t \equiv R$ and 
	$\beta_t \equiv \beta$ are positive constants with $\beta R < 1$ and 
	$Y_t$ has a positive lower bound. Without preference shocks, 
	\citet*{ma2021theory} show that the asymptotic MPC is 
	\begin{equation*}
		\bar{c}(z) = \begin{cases}
			0, & \text{if $\beta R^{1-\gamma} \geq 1$,}  \\
			1-(\beta R^{1-\gamma})^{1/\gamma}, 
			& \text{otherwise.}   
		\end{cases}
	\end{equation*}
	With preference shocks, Theorems~\ref{t:0ampc_thm}--\ref{t:non0_ampc3} 
	imply that the asymptotic MPC is
	\begin{equation*}
		\bar c(z_{ij}) = \begin{cases}
			0, & \text{if $\sum_{\ell=1}^{i-1} \bar p_{i\ell} > 0$
				or $\bar p_{ii} \beta R^{1-\gamma_i} \geq 1$,}\\
			1 - (\bar p_{ii} \beta R^{1-\gamma_i})^{1/\gamma_i}, 
			& \text{otherwise.}
		\end{cases}
	\end{equation*}	
	In particular, a positive probability of transitioning to a lower risk aversion state is sufficient to generate vanishing asymptotic MPCs, regardless of whether the classical condition $\beta R^{1-\gamma}\ge 1$ holds. Thus, the source of vanishing MPCs shifts from \textit{restrictive parameter configurations} to \textit{the dynamics of risk preferences themselves}.
	
	At the opposite extreme, if $\sum_{\ell=1}^{i}\bar p_{i\ell}=0$, then by Theorem~\ref{t:non0_ampc2} and Proposition~\ref{pr:non-concave}, the asymptotic MPC equals one and the optimal consumption policy is nonconcave in wealth.
\end{example}

\begin{example}[Optimal savings with {\sc iid} return and discounting]
	\label{ex:iid_ops}
	Suppose, as in \cite{benhabib2015wealth}, $\{\beta_{t}\}$ and $\{R_{t}\}$ 
	are {\sc iid} driven by innovation $\{\epsilon_t\}$ with positive lower 
	bounds, and $\{Y_t\}$ has a positive lower bound. 
	Without preference shocks, \cite{ma2021theory} show that
	\begin{equation*}
		\bar c(z) =\begin{cases}
			0, & \text{if } \EE \beta R^{1-\gamma}\geq 1 \\
			1 - (\EE \beta R^{1-\gamma})^{1/\gamma}, & \text{otherwise.}
		\end{cases}
	\end{equation*}
	With preference shocks, Theorems~\ref{t:0ampc_thm}--\ref{t:non0_ampc3} 
	imply that the asymptotic MPC is
	\begin{equation*}
		\bar c(z_{ij}) = \begin{cases}
			0, & \text{if $\sum_{\ell=1}^{i-1} \bar p_{i\ell} > 0$ or 
				$\bar p_{ii} \EE \beta R^{1-\gamma_i} \geq 1$,} \\
			1 - (\bar p_{ii} \EE \beta R^{1-\gamma_i})^{1/\gamma_i},
			& \text{otherwise.}
		\end{cases}
	\end{equation*}
	As in Example~\ref{ex:clas_ops}, vanishing asymptotic MPCs no longer require the classical condition $\EE \beta R^{1-\gamma}\ge 1$. Instead, they emerge whenever agents face a positive probability of becoming less risk averse in the future. Similarly, the asymptotic MPC converges to 
	one and the optimal consumption function is not concave in wealth at state 
	$\gamma_i$ when $\sum_{\ell=1}^{i} \bar p_{i\ell} = 0$. Therefore, risk preference dynamics become the primary driver of asymptotic saving behavior.
\end{example}

Together, these examples highlight that risk preference dynamics constitute a new and powerful determinant of long-run saving behavior.

\subsection{The Asymptotic Saving Rate}\label{ss:a_saving}

As mentioned earlier, we follow the literature and define an agent’s 
\textit{saving rate} as the change in net worth relative to total income, 
excluding capital losses (see the discussion in Footnote~\ref{fn:saving_rate}). By \eqref{eq:trans_at} and some simple algebra, the saving rate can be expressed as
\begin{equation*}
	s_{t+1} = 1 - \frac{(R_{t+1}-1)^-(1-c_t/w_t) + c_t/w_t
	}{(R_{t+1}-1)^+(1-c_t/w_t)+Y_{t+1}/w_t} \in (-\infty, 1),
\end{equation*}
where $x^+ = \max\{x,0\}$ and $x^- = -\min\{x,0\}$ for each $x\in \RR$.
Taking the limit as $w_t \to \infty$, we obtain the \textit{asymptotic saving rate}
\begin{equation}\label{eq:a_sr}
	\bar s := 1 - \frac{(\hat R -1)^-(1-\bar c) + \bar c
	}{(\hat R - 1)^+(1-\bar c)} \in [-\infty, 1],
\end{equation}
where $\bar c$ denotes the asymptotic MPC. 

Equation \eqref{eq:a_sr} highlights the intimate link between asymptotic MPCs and asymptotic saving rates. In particular, if wealth yields a positive gross return (i.e., $R_t>1$) and the asymptotic MPC is zero, the agent's asymptotic saving rate reaches 100\%. Conversely, if the asymptotic MPC is positive and sufficiently large, the asymptotic saving rate can become negative, reflecting that precautionary saving motives eventually vanish as wealth grows without bound. Thus, the asymptotic saving rate directly inherits the qualitative features of the asymptotic MPC, providing an intuitive metric for understanding long-run wealth accumulation under stochastic risk preferences.

\section{Quantitative Analysis}\label{s:quant}

In this section, we complement the theory with a quantitative analysis that serves two purposes. First, using a parsimonious quantitative specification, we provide numerical illustrations that render the main theoretical results transparent and economically intuitive. Second, within a more general quantitative environment, we compute generalized impulse response functions to elucidate the economic mechanisms through which preference shocks generate qualitative changes in model dynamics.


\subsection{Model Specification}\label{ss:quant_specif}

Consider an optimal savings problem with risk aversion shocks. The discount factor and the gross rate of return on wealth are held constant over time with $\beta_t \equiv \beta$ and $R_t \equiv R$. Risk aversion follows a discretized AR(1) process:
\begin{equation*}
	\gamma_{t+1} = \mu_\gamma(1-\rho_\gamma) + \rho_\gamma \gamma_t 
	+ \sigma_\gamma \sqrt{1-\rho_\gamma^2} \epsilon_{t+1}^\gamma,
	\quad (\epsilon_t^\gamma)\iidsim N(0,1).
\end{equation*}
This specification is a special case of our theory, with the preference state given by $\bar Z_t = \gamma_t$. For nonfinancial income, we allow for alternative specifications depending on the quantitative experiment under consideration. The model is solved using the endogenous grid method of \citet{carroll2006method}. Detailed implementation and algorithmic steps are provided in the Online Appendix.

In our quantitative experiments, we work at an annual frequency and set $\beta = 0.95$ and $R = 1.02$. In our baseline setting, risk aversion is parameterized by $\rho_\gamma = 0.6$, $\mu_\gamma = 4$, and $\sigma_\gamma = 1$, and discretized into a finite-state Markov chain using the method of \citet{tauchen1986finite}. The number of discrete states is allowed to vary across experiments. In particular, for the impulse response analysis in Section~\ref{ss:quant_irf}, we employ a fine discretization to improve approximation accuracy. In all specifications, the minimum, median (and mean), and maximum values of $\gamma_t$ are $1$, $4$, and $7$, respectively, which lie well within the range commonly used in empirical research.

\subsection{The Asymptotic Saving Rates}\label{ss:quant_asym}

To better illustrate the implications of our theory, we begin by comparing our model with the standard optimal savings problem without preference shocks. To this end, we consider simply $Y_t \equiv 1$ and discretize $\{\gamma_t\}$ into a three-state Markov chain. The resulting state space and transition matrix are 
\begin{equation*}
	\gamma_t \in \{1, 4, 7\}
	\quad \text{and} \quad 
	\bar P = \begin{bmatrix}
		0.6416 & 0.3538 & 0.0001 \\
		0.0303 & 0.9392 & 0.0305 \\
		0.0001 & 0.3538 & 0.6461
	\end{bmatrix}.
\end{equation*}
The stationary mean of this discrete process coincides with the median state $\gamma_t = 4$. At this state, the probability of remaining at $\gamma_t = 4$ next period is 93.92\%, while the probability of transitioning to the lower risk aversion state $\gamma_t = 1$ is only 3.03\%.

To ensure a fair comparison, we contrast the optimal consumption function when $\gamma_t$ is fixed at its stationary mean $\EE \gamma_t = 4$ in our stochastic risk aversion model with that of the standard optimal savings problem featuring CRRA utility with a constant risk aversion coefficient $\gamma = 4$. Figure~\ref{fig:cons} plots the results.\footnote{The 
	main difference between the left and right panels lies in the range of wealth displayed. Moreover, because the right panel is set on a log scale, the consumption function appears nonconcave, though it is in fact concave. A log scale is similarly used in Figure~\ref{fig:MPC_sr} when required.}
Despite the fact that current risk aversion is identical across the two models and that there is a high probability of remaining at the same level of risk aversion next period, consumption is substantially lower under stochastic risk aversion (red dotted line) than under constant risk aversion (blue solid line), reflecting the precautionary behavior induced by the possibility of future changes in risk aversion.

Figure~\ref{fig:MPC_sr} plots the consumption rate, defined as $c^*(w,z)/w$, and the saving rate as a function of wealth. The left panel shows that, in the standard optimal savings model, the consumption rate (blue solid line) converges to the theoretical asymptotic MPC, which is strictly positive and equal to $1-(\beta R^{1-\gamma})^{1/\gamma} = 0.0273$ as derived in Example~\ref{ex:clas_ops}. By contrast, in the stochastic risk aversion model, the consumption rate exhibits a pronounced downward trend and continues to decline even at very high levels of wealth. This pattern is fully consistent with Theorem~\ref{t:0ampc_thm}, which establishes that the asymptotic MPC is equal to zero.

\begin{figure}
	\begin{subfigure}{.45\linewidth}
		\centering
		\includegraphics[width=\linewidth]{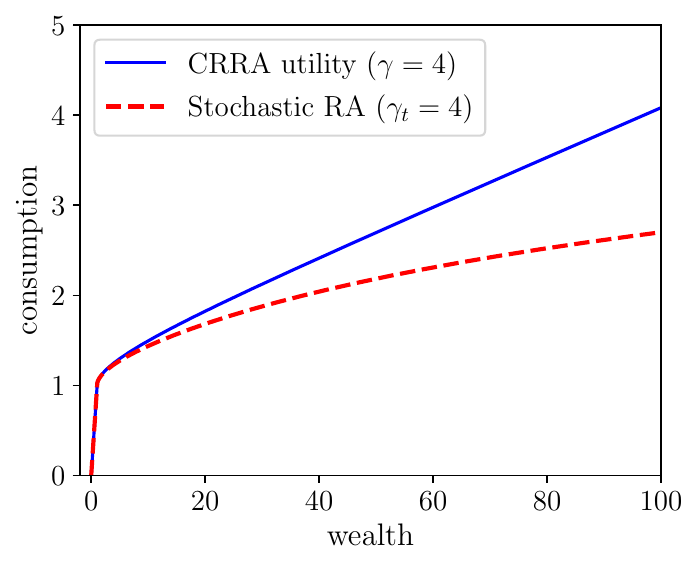}
	\end{subfigure}%
	\begin{subfigure}{.45\linewidth}
		\centering
		\includegraphics[width=\linewidth]{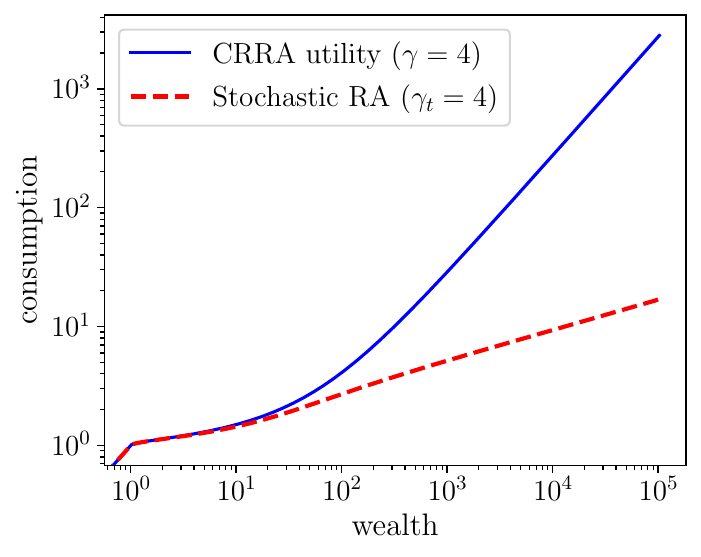}
	\end{subfigure}
	\caption{The optimal consumption function}\label{fig:cons}
\end{figure}

\begin{figure}
	\begin{subfigure}{.45\linewidth}
		\centering
		\includegraphics[width=\linewidth]{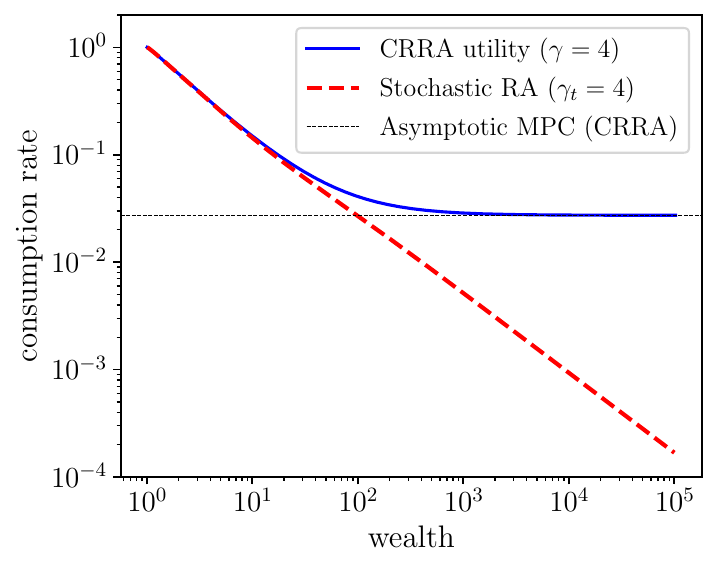}
	\end{subfigure}%
	\begin{subfigure}{.45\linewidth}
		\centering
		\includegraphics[width=\linewidth]{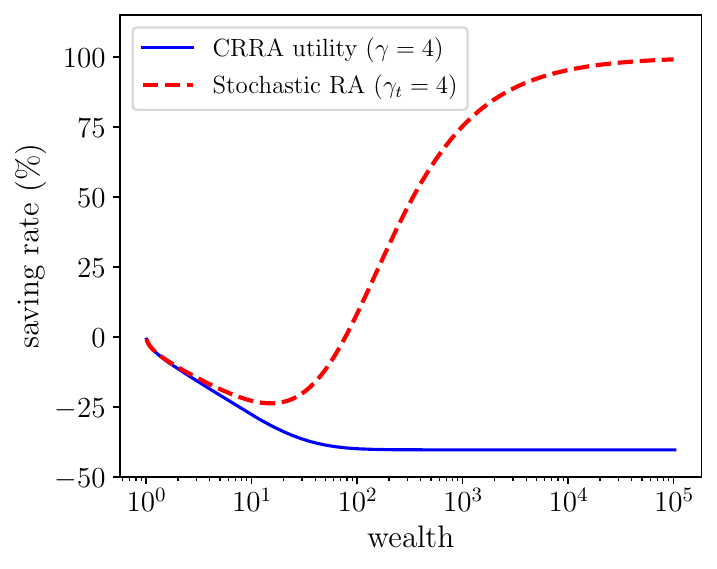}
	\end{subfigure}
	\caption{The consumption and saving rates}\label{fig:MPC_sr}
\end{figure}

The right panel plots the corresponding saving rates. In the standard optimal savings model, saving rates are negative and converge to their theoretical asymptotic value of $-40.31\%$, as given by \eqref{eq:a_sr}, in stark contradiction with the empirical evidence on wealthy households. In contrast, under stochastic risk aversion, the saving rate converges to its theoretical limit of $100\%$. Taken together, these results illustrate that even minor probabilities of shifts in risk aversion can have sizable effects on the saving behavior, suggesting that stochastic risk aversion offers a promising mechanism for generating persistently high saving rates among wealthy households, in line with the patterns observed in the data.

\subsection{Impulse Response Analysis}\label{ss:quant_irf}

In this section, we systematically examine how changes in risk aversion influence consumption and saving behavior using impulse response functions (IRFs). To capture the model’s nonlinear dynamics, we construct IRFs as state-dependent random variables, following  \citet{koop1996impulse}. These IRFs allow us to trace how shocks to risk aversion propagate over time and affect consumption, saving rates, and consumption volatility across different states.\footnote{Consumption 
	volatility in period-$t$ is defined as the conditional variance: 
	\begin{equation*}
		\Var_{t-1} [c^*(w_t,Z_t)] := 
		\EE_{t-1} [c^*(w_t,Z_t)]^2 - [\EE_{t-1} c^*(w_t,Z_t)]^2.
	\end{equation*}
}

We follow the standard earnings specification in the literature and assume that labor income evolves according to
\begin{align*}
	&\log Y_t = \eta_t + \nu_t, \quad (\nu_t) \sim N(0, \sigma_v^2), \\
	&\eta_{t+1} =  \rho_\eta \eta_t + \epsilon_t^\eta, 
	\quad (\epsilon_t^\eta) \sim N(0, \sigma_\eta^2).
\end{align*}
Consistent with the empirical estimates reported in \citet{lander2024estimating}, we set $\rho_\eta = 0.9735$, $\sigma_\eta = 0.1272$, and $\sigma_\nu = 0.2102$, based on PSID data over 1970-2012. 

To ensure sufficient numerical accuracy, we discretize $\{\gamma_t\}$ into a 121-state Markov chain following \citet{tauchen1986finite}. In addition, the persistent income component $\{\eta_t\}$ is discretized into 5 states, while expectations over the transitory shock $\{\nu_t\}$ are computed using a 7-point Gauss–Hermite quadrature. 
Sensitivity analysis confirms that the results are robust to parametric choices. 

\begin{figure}
	\includegraphics[width=.987\linewidth]{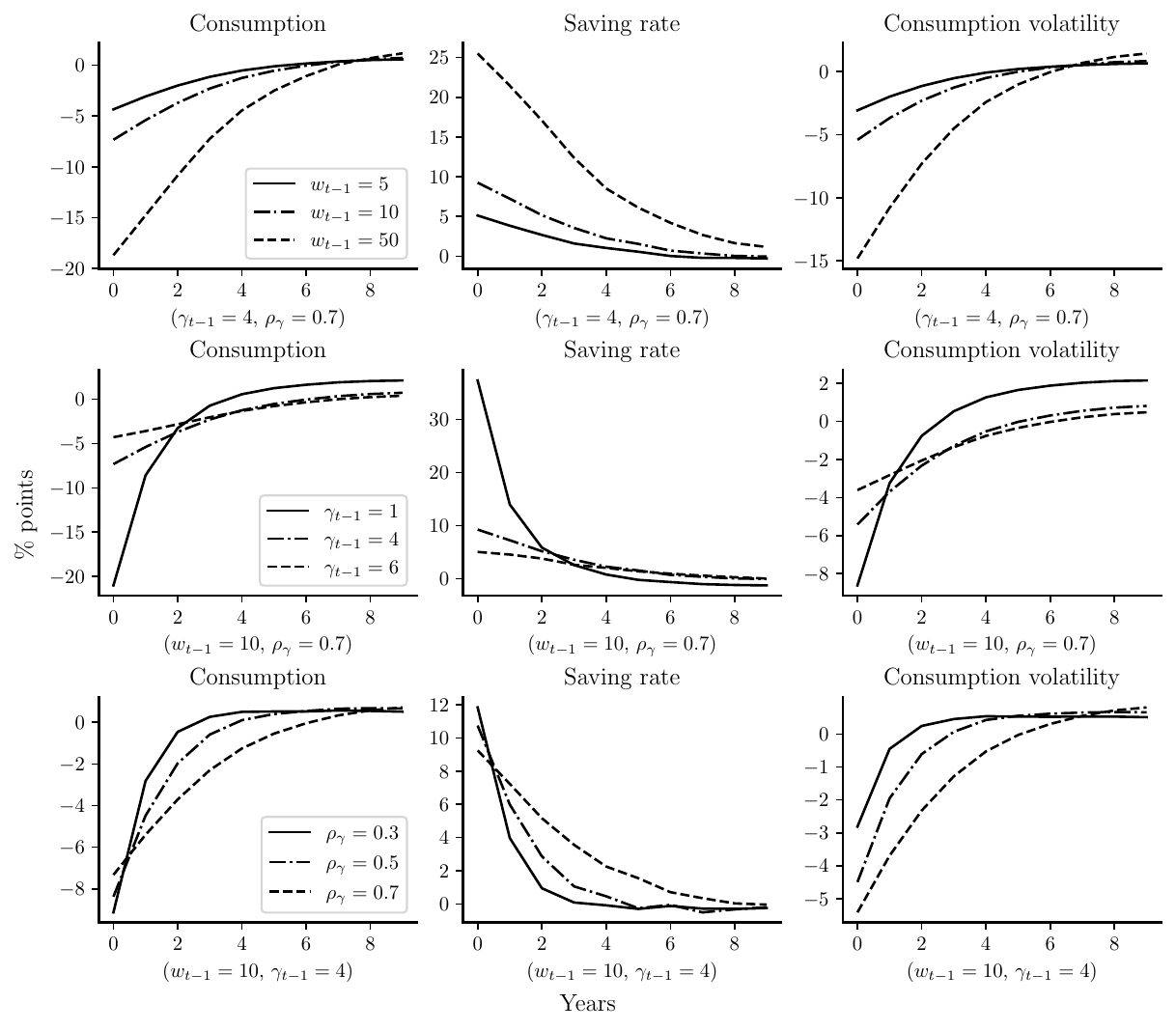}
	\caption{Impulse responses to a positive one-unit risk aversion shock}\label{fig:irf}
\end{figure}

Figure~\ref{fig:irf} reports impulse response functions to a one-unit increase in risk aversion at the beginning of period $t$. As noted earlier, the IRFs are computed as state-dependent responses, conditional on different initial levels of wealth, risk aversion, and the persistent income component $(w_{t-1},\gamma_{t-1},\eta_{t-1})$. All experiments are conducted conditional on $\eta_{t-1} = 0$, its stationary mean.

The top row of Figure~\ref{fig:irf} illustrates heterogeneity with respect to initial wealth, holding risk aversion fixed at its stationary mean $\gamma_{t-1} = 4$. The top-left panel shows that an increase in risk aversion induces an immediate and persistent decline in consumption, with the magnitude of the response rising sharply in wealth. When wealth rises from $w_{t-1} = 5$ to $10$ and $50$, contemporaneous consumption falls by 4.67\%, 7.89\%, and 20.23\%, respectively, with effects lasting roughly 2–6 years. This nonlinearity reflects the fact that low-wealth households are constrained in their ability to adjust savings, so consumption responds weakly to preference shocks, whereas wealthy households can reallocate intertemporal resources more flexibly.

The top-middle panel shows that this adjustment operates primarily through saving. A one-unit increase in risk aversion raises the saving rate by 5.81\%, 10.15\%, and 28.05\% at $w_{t-1} = 5, 10,$ and $50$, respectively, with effects persisting for 4–8 years. Thus, as wealth rises, the response shifts from limited consumption adjustment toward substantial increases in saving. This pattern aligns with our theoretical mechanism: at high wealth levels, precautionary motives dominate, prompting agents to accumulate assets aggressively even in response to transitory risk aversion shocks. The top-right panel further shows that consumption volatility declines following a positive risk aversion shock, particularly for wealthier households, reflecting stronger consumption smoothing driven by heightened precautionary behavior.

The middle row of Figure~\ref{fig:irf} illustrates heterogeneity with respect to initial risk aversion, holding initial wealth fixed at $w_{t-1} = 10$. When initial risk aversion is high ($\gamma_{t-1} = 6$), the same shock generates relatively modest responses: consumption falls by 4.84\% and the saving rate rises by 5.53\%. By contrast, when agents start from a low risk aversion state ($\gamma_{t-1} = 1$), consumption drops sharply by 19.26\% and saving increases by 32.21\%. This asymmetry reflects state-dependent amplification: a positive shock represents a larger proportional change in effective prudence when baseline risk aversion is lower, triggering a stronger reoptimization of intertemporal choices. Consistent with this mechanism, the decline in consumption volatility is also largest at low initial risk aversion.

The bottom row of Figure~\ref{fig:irf} examines the role of the persistence of risk aversion shocks, governed by $\rho_\gamma$, holding initial wealth and risk aversion fixed at $(w_{t-1},\gamma_{t-1})=(10,4)$. Higher persistence markedly amplifies and prolongs the responses of consumption and saving. As $\rho_\gamma$ increases, households cut consumption more sharply, raise saving more persistently, and experience larger declines in consumption volatility. Intuitively, greater persistence raises the expected duration of elevated risk aversion, strengthening precautionary motives and leading agents to treat preference shocks as quasi-permanent. This persistence channel thus reinforces the state-dependent and nonlinear dynamics emphasized by our theory.

Taken together, these impulse responses deliver three key insights. First, shocks to risk aversion generate economically large and persistent movements in consumption and saving. Second, these effects are highly state dependent, varying systematically with households’ wealth, the current level of risk aversion, and the persistence of preference shocks. Third, the IRFs provide a quantitative illustration of the core theoretical mechanism: stochastic risk preferences alone can fundamentally reshape saving behavior, even in environments with constant returns and discounting, accounting for the high and persistent saving rates observed among wealthy households.

\appendix

\section{Preliminaries}

Let $(\Omega, \fF, \PP)$ be a fixed probability space on which all random 
variables are defined, and let $\EE$ denote expectations under $\PP$. The state 
process $\{Z_t\}$ and the innovation process $\{\epsilon_t\}$ introduced in 
\eqref{eq:law_exog} are defined on this space, with $\pi$ denoting the marginal 
distribution of $\{\epsilon_t\}$. Let $\{\fF_t\}$ be the natural filtration 
generated by $\{Z_t\}$ and $\{\epsilon_t\}$. For $z\in \ZZ$, we write $\PP_z$ 
for probability conditional on $Z_0 = z$ and $\EE_z$ for expectation under 
$\PP_z$.

For a stochastic process $\phi_t = \phi(Z_{t-1}, Z_t, \epsilon_t)$, where 
$\phi$ is a nonnegative measurable function, we define the matrix $L_\phi$ as
\begin{equation}\label{eq:lfunc}
	L_\phi(z, \hat z) := 
	P(z, \hat z) \int \phi(z, \hat z, \hat \epsilon) \pi (\diff \hat \epsilon). 
\end{equation}
The matrix $L_\phi$ is expressed as a function on $\ZZ \times \ZZ$ in 
\eqref{eq:lfunc}, but it can be represented in traditional matrix notation by 
enumerating $\ZZ$. Recall that $\RR^{\ZZ}$ is the set of real-valued functions 
on $\ZZ$. In what follows, we treat $L_\phi$ as a linear operator on 
$\RR^{\ZZ}$. A proof by induction shows that, for any $h \in \RR^{\ZZ}$, we have
\begin{equation}
	\label{eq:ln_ope}
	(L_\phi^n \, h)(z) 
	= \EE_z \prod_{t=1}^n \phi_t h(Z_t),
\end{equation}
where $L_\phi^n$ is the $n$-th composition of the operator $L_\phi$ with
itself, or equivalently, the $n$-th power of the matrix $L_\phi$. The following result holds. The proof is similar to the proof of Lemma~4.1 of \cite{ma2020income} and thus omitted.

\begin{lemma}\label{lm:Lphi}
	Let $\{\phi_t \}$ and $L_\phi$ be as defined in \eqref{eq:lfunc}. If
	$r(L_\phi) < 1$, then there exist $N \in \NN$ and $\sigma \in (r(L_\phi), 1)$ such that $\max_{z \in \ZZ} \EE_z \prod_{t=1}^n \phi_t < \sigma^n$ for all $n \geq N$.
\end{lemma}


\section{Proof of Section~\ref{s:or} Results}
\label{s:proof_or}

Throughout this section, we assume that
Assumptions~\ref{a:utility}--\ref{a:spec_and_finite_exp} 
hold. Without loss of generality, we may assume $\ZZ = \{1, \dots, Z\}$ in the 
related proofs of this section.

\begin{proposition}
	\label{pr:complete}
	$(\cC, \rho)$ is a complete metric space.
\end{proposition}

\begin{proof}
	The proof is a straightforward extension of Proposition~4.1 of \cite{li2014solving} and thus omitted. 
	A full proof is available from the authors upon request.
\end{proof}

\begin{proposition}\label{pr:welldef_T}
	For all $c \in \cC$ and $(w,z) \in \SS$, there exists a unique $\xi \in (0,w]$ that solves \eqref{eq:tio}.
\end{proposition}

\begin{proof}
	For given $c \in \cC$, we can rewrite \eqref{eq:tio} as 
	\begin{equation}\label{eq:psi_c}
		u'(\xi, z) 
		= \psi_c(\xi, w,z) := \max \{g_c(\xi, w,z), u'(w, z)\},
	\end{equation}
	where $g_c$ is a function on 
	\begin{equation}\label{eq:G_set}
		G:= \{(\xi,w,z)\in \RR_+\times \SS: 0 < \xi \leq w \}
	\end{equation}
	defined by
	\begin{equation}\label{eq:g_c}
		g_c(\xi,w,z) := \EE_z \hat \beta \hat R u'\left(
		    c(\hat R(w-\xi) + \hat Y, \hat Z), \hat Z
		\right).
	\end{equation}
	Fix $c \in \cC$ and $(w,z) \in \SS$. Since $c \in \cC$, the map 
	$\xi \mapsto \psi_c(\xi,w,z)$ is increasing. Additionally, because 
	$\xi \mapsto u'(\xi, z)$ is strictly decreasing, \eqref{eq:tio} can 
	have at most one solution. Therefore, uniqueness holds.
	
	To establish existence, we apply the intermediate value theorem. It 
	suffices to verify the following three conditions:
	\begin{enumerate}
		\item[(a)] the map $\xi \mapsto \psi_c(\xi, w, z)$ is continuous on $(0, w]$,
		\item[(b)] there exists a $\xi \in (0,w]$ such that 
		$u'(\xi, z) \geq \psi_c(\xi, w, z)$, and
		\item[(c)] there exists a $\xi \in (0,w]$ such that 
		$u'(\xi, z) \leq \psi_c(\xi, w, z)$.
	\end{enumerate}
	For part (a), it is sufficient to show that $\xi \mapsto g_c(\xi,w,z)$ is 
	continuous on $(0,w]$. To this end, fix $\xi \in (0,w]$ and let $\xi_n \to 
	\xi$. Since $c\in \cC$, there exists a constant $M \in \RR_+$ such that 
	\begin{equation*}
		u'(w,z) \leq u'(c(w,z),z) \leq u'(w,z) + M
		\quad \text{for all } (w,z) \in \SS.
	\end{equation*}
	By the monotonicity of $u'$, we have
	\begin{equation}\label{eq:uppbd_ruprmc}
		\hat{\beta} \hat{R} 
		    u'[c (\hat{R} \left(w - \xi \right) + \hat{Y}, \hat{Z}),
		    \hat Z]
		\leq \hat{\beta} \hat{R} 
		    u'[c(\hat{Y}, \hat{Z}), \hat Z]
		\leq \hat{\beta} \hat{R} u'( \hat{Y}, \hat Z) 
		    + \hat{\beta} \hat{R} M.
	\end{equation}
	Furthermore, by Assumption~\ref{a:spec_and_finite_exp}, we know that, 
	for all $z\in \ZZ$,
	\begin{equation*}
		\EE_z \hat \beta \hat R u'(\hat Y, \hat Z) < \infty
		\quad \text{and} \quad
		\EE_z \hat \beta \hat R < \infty.
	\end{equation*}
	Hence, by the dominated convergence theorem and the continuity of $c$, we 
	can conclude that $g(\xi_n,w,z) \to g(\xi,w,z)$. This proves that 
	$\xi \mapsto \psi_c(\xi, w, z)$ is continuous.
	
	Part~(b) clearly holds, since $u'(\xi,z) \to \infty$ as $\xi \to 0$ 
	and $\xi \mapsto \psi_c(\xi, w, z)$ is increasing and always finite (since 
	it is continuous as shown in the previous paragraph). Part~(c) is also 
	trivial (just set $\xi = w$).
\end{proof}

\begin{proposition}\label{pr:self_map}
	We have $Tc \in \cC$ for all $c \in \cC$.
\end{proposition}

\begin{proof}
	Fix $c \in \cC$ and let $\psi_c$ and $g_c$ be defined as in \eqref{eq:psi_c} and \eqref{eq:g_c}, respectively.
	
	\textbf{Step~1.} We show that $Tc$ is continuous. To apply a standard 
	fixed point parametric continuity result, such as Theorem~B.1.4 of 
	\cite{stachurski2009economic}, we first demonstrate that $\psi_c$ is jointly 
	continuous on the set $G$ defined in \eqref{eq:G_set}. This will hold
	if $g_c$ is jointly continuous on $G$. For any $\{ (\xi_n, w_n, z_n) \}$ 
	and $(\xi, w, z)$ in $G$ with $(\xi_n, w_n, z_n) \to (\xi, w, z)$, we 
	need to show that $g_c(\xi_n, w_n, z_n) \to g_c(\xi, w, z)$. To this 
	end, we define
	\begin{align*}
		h_1 ( \xi, w, Z, \hat{Z}, \hat{\epsilon}), \,
		h_2 ( \xi, w, Z, \hat{Z}, \hat{\epsilon})  
		:= \hat{\beta} \hat{R} \left(
		    u'(\hat Y,\hat Z) + M \pm 
		    u'(c(\hat{R} \left( w - \xi \right) + \hat{Y}, \hat{Z}), \hat Z) 
		\right),
	\end{align*}
	where $\hat{\beta} := \beta (Z, \hat{Z}, \hat{\epsilon})$, 
	$\hat{R} := R (Z, \hat{Z}, \hat{\epsilon})$ and 
	$\hat{Y} := Y (Z, \hat{Z}, \hat{\epsilon})$ as defined in 
	\eqref{eq:law_exog}. Then $h_1$ and $h_2$ are continuous in 
	$(\xi, w, Z, \hat{Z})$ by the continuity of $c$  and nonnegative by 
	\eqref{eq:uppbd_ruprmc}.
	
	By Fatou's lemma and Theorem~1.1 of \cite{feinberg2014fatou}, we have
	\begin{align*}
		\int \sum_{\hat z \in \ZZ}
		h_i ( \xi, w, z, \hat{z}, \hat{\epsilon})
		P(z,\hat{z}) \pi(\diff \hat{\epsilon}) 
		&\leq \int \liminf_{n \to \infty}
		\sum_{\hat z \in \ZZ}
		h_i ( \xi_n, w_n, z_n, \hat{z}, \hat{\epsilon})
		P(z_n, \hat{z})
		\pi(\diff \hat{\epsilon})  \\
		& \leq \liminf_{n \to \infty}
		\int 
		\sum_{\hat z \in \ZZ}
		h_i ( \xi_n, w_n, z_n, \hat{z}, \hat{\epsilon})
		P( z_n, \hat z)
		\pi (\diff \hat{\epsilon}),
	\end{align*}
	which implies 
	\begin{equation*}
		\liminf_{n \to \infty} \left[
		\pm \EE_{z_n} \hat{\beta} \hat{R} 
		u'(c(\hat{R} \left(w_n - \xi_n \right) + \hat{Y}, \hat{Z}),
		    \hat Z)
		\right]
		\geq \left[ 
		\pm \EE_{z} \hat{\beta} \hat{R}
		u'(c(\hat{R} \left(w - \xi \right) + \hat{Y}, \hat{Z}), 
		    \hat Z)
		\right].
	\end{equation*}
	This shows that $g_c$ is then continuous, since the inequality above is 
	equivalent to the statement
	\begin{equation*}
		\liminf_{n \to \infty} g_c(\xi_n, w_n, z_n)   
		\geq g_c(\xi, w, z) \geq   
		\limsup_{n \to \infty} g_c(\xi_n, w_n, z_n).
	\end{equation*}
	Hence, $\psi_c$ is continuous on $G$, as required. Moreover, since $\xi 
	\mapsto \psi_c(\xi, w, z)$ takes values in the closed interval 
	\begin{equation*}
		I(w,z) := \left[
		    u'(w,z), 
		    u'(w,z) + 
		    \EE_z \hat{\beta} \hat{R} (u'(\hat{Y}, \hat Z) + M)
		\right]
	\end{equation*}
	and the correspondence $(w, z) \mapsto I(w,z)$ is nonempty, compact-valued 
	and continuous, Theorem~B.1.4 of \cite{stachurski2009economic} implies that
	$Tc$ is continuous on $\SS$.
	
	\textbf{Step 2.} We show that $Tc$ is increasing in $w$. Suppose that for 
	some $z \in \ZZ$ and $w_1, w_2 \in (0, \infty)$ with $w_1 < w_2$, we have 
	$\xi_1 := Tc (w_1,z) > Tc (w_2,z) =: \xi_2$. Since $c$ is increasing in $w$ 
	by assumption, $\psi_c$ is increasing in $\xi$ and decreasing in $w$. Thus, 
	we have $u'(\xi_1,z) < u'(\xi_2,z) = \psi_c(\xi_2, w_2, z) \leq \psi_c(\xi_1, w_1, z) = u'(\xi_1,z)$, which is a contradiction.
	
	\textbf{Step 3.} We have shown in Proposition~\ref{pr:welldef_T} that $Tc(w,z) \in (0,w]$ for all $(w,z) \in \SS$.
	
	\textbf{Step 4.} Since $u'[Tc(w,z),z] \geq u'(w, z)$, we have
	\begin{align*}
		&\left| u'[Tc(w,z), z] - u'(w, z) \right| 
		= u'[Tc(w,z),z] - u'(w,z)    \\
		& \leq \EE_{z} \hat{\beta} \hat{R} 
		u'\,[c (\hat{R}( w - Tc(w,z)) + \hat{Y}, \, \hat{Z}), 
		    \hat Z \,]
		\leq \EE_{z} \hat{\beta} \hat{R} (u'(\hat{Y},\hat Z) + M)
	\end{align*}
	for all $(w,z) \in \SS$. The last term is finite by 
	Assumption~\ref{a:spec_and_finite_exp}.
\end{proof}

For each $h:(0,\infty) \to \RR_+^{\ZZ}$, we write 
\begin{equation*}
    h(w) = (h_1(w), \dots, h_Z(w)).	
\end{equation*}
To prove Theorem~\ref{t:opt}, let $\hH$ denote 
the set of all continuous functions $h:(0,\infty) \to \RR_+^{\ZZ}$ such that 
each $h_z$ is decreasing and $w \mapsto h_z(w)-u'(w,z)$ is bounded and 
nonnegative. For any $h \in \hH$, we define $(\tilde{T} h)_z(w)$ as the value 
$\kappa$ that solves
\begin{equation}\label{eq:kappa}
	\kappa = \max \left\{ 
	    \EE_{z} \, \hat{\beta} \hat{R} \,
	    h (\hat{R} \,[w - (u'(\cdot,z))^{-1}(\kappa)] + \hat{Y}, \hat{Z}), 
	    u'(w,z) 
	\right\},
\end{equation}
where, with a slight abuse of notation, we denote $h_z(w) := h(w,z)$.
Moreover, we consider the bijection $H: \cC \to \hH$ defined by 
$(Hc)_z(w) := u'(c(w,z),z)$. The next lemma implies that $\tilde{T}$ is 
a well-defined self-map on $\hH$, as well as topologically conjugate to $T$.

\begin{lemma}\label{lm:conjug}
	The operator $\tilde{T} \colon \hH \to \hH$ and satisfies $\tilde{T} H = H 
	T$ on $\cC$.
\end{lemma}

\begin{proof}
	Pick any $c \in \cC$ and $(w,z) \in \SS$. Let $\xi := Tc(w,z)$. By definition, $\xi$ solves
	\begin{equation}
		\label{eq:Tc_eq}
		u'(\xi, z) = 
		\max \left\{ 
			\EE_z \hat{\beta} \hat{R}
			u'(c(\hat{R} \left(w - \xi \right) + \hat{Y}, \hat{Z}), 
			\hat Z), u'(w, z) 
		\right\}.
	\end{equation}
	We need to show that $HTc$ and $\tilde{T} Hc$ evaluate to the same number 
	at $(w,z)$. In other words, we need to verify that $u'(\xi,z)$ is the 
	solution to
	\begin{equation*}
		\kappa = \max \left\{ 
			\EE_{z} \hat \beta \hat R 
			u'(c(\hat{R} \, [w - u'(\cdot,z)^{-1} (\kappa)] + \hat Y, \hat{Z}), 
			\hat Z), u'(w,z) 
		\right\}.
	\end{equation*}
	But this follows immediately from \eqref{eq:Tc_eq}. Hence, we have shown that
	$\tilde{T} H = H T$ on $\cC$. Since $H \colon \cC \to \hH$ is a bijection,
	we have $\tilde{T} = HT H^{-1}$. Moreover, Proposition~\ref{pr:self_map} 
	ensures that $T \colon \cC \to \cC$, and hence $\tilde{T} \colon \hH \to \hH$. 
	This completes the proof.
\end{proof}

To prove Theorem~\ref{t:opt}, we apply the Perov contraction theorem following \cite{toda2025essential}. For $h,g\in \hH$ and $z\in \ZZ$, define
\begin{equation*}
	d_z(h,g):= \sup_{w\in (0,\infty)}|h_z(w)-g_z(w)|.
\end{equation*}
Note that $d_z$ is always finite by the definition of $\hH$. Define the vector-valued metric $d:\hH\times \hH\to \RR_+^{\ZZ}$ by $d(h,g)=(d_1(h,g),\dots,d_Z(h,g))$. By Proposition~\ref{pr:complete} and the discussion in Section~7.5 of \cite{toda2025essential}, $(\hH,d)$ is a complete vector-valued metric space.
\begin{lemma}\label{lm:Perov}
	Let $B:= K(1)$. Then $\tilde{T}:\hH\to \hH$ is a Perov contraction with coefficient matrix $B$, that is, for any $h_1,h_2\in \hH$, we have $d(\tilde{T}h_1,\tilde{T}h_2)\le Bd(h_1,h_2)$.
\end{lemma}

\begin{proof}
	By Assumption~\ref{a:spec_and_finite_exp}~\eqref{ac:spec}, we have 
	$r(B)\in [0,1)$. By Theorem~7.7 of \cite{toda2025essential}, 
	it suffices to verify the following two conditions:
	\begin{enumerate}
		\item\label{item:Perov_a} 
		If $h_1,h_2\in \hH$ and $h_1\le h_2$ pointwise, then $\tilde{T}h_1\le \tilde{T}h_2$ pointwise.
		\item\label{item:Perov_b} 
		For any $h\in \hH$ and $a\in \RR_+^{\ZZ}$, we have $\tilde{T}(h+a)\le \tilde{T}h+Ba$ pointwise.
	\end{enumerate}
	\eqref{item:Perov_a} Suppose $h_1\le h_2$ and take any $(w,z)\in \SS$. For $i\in\{1,2\}$, let $\kappa_i=(\tilde{T}h_i)_z(w)$, which satisfies \eqref{eq:kappa}. Let $\xi_i=u'(\cdot,z)^{-1}(\kappa_i)$. Since $u'(\cdot,z)$ is strictly decreasing, to prove $\kappa_1\le \kappa_2$, it suffices to show $\xi_1\ge \xi_2$. Suppose on the contrary that $\xi_1<\xi_2$. By Lemma \ref{lm:conjug} and \eqref{eq:kappa}, we obtain
	\begin{align*}
		u'(\xi_1, z)
		&=\max\{\EE_z \hat{\beta}\hat{R} \,
		    h_1(\hat{R}(w-\xi_1)+\hat{Y},\hat{Z}),u'(w, z) \} \\
		&\leq \max\{\EE_z \hat{\beta}\hat{R} \,
		    h_2(\hat{R}(w-\xi_1)+\hat{Y},\hat{Z}),u'(w, z) \} \\
		&\leq \max\{\EE_z \hat{\beta}\hat{R} \, 
		    h_2(\hat{R}(w-\xi_2)+\hat{Y},\hat{Z}),u'(w, z) \}
		 = u'(\xi_2, z),
	\end{align*}
	where the first inequality uses $h_1\le h_2$ and the second inequality uses 
	the fact that $h_2$ is decreasing in $w$ and $\xi_1<\xi_2$. Since $u'$ is 
	strictly decreasing, we obtain $\xi_1\geq \xi_2$, which is a contradiction.
	
	\eqref{item:Perov_b} Let $h\in \hH$ and $a\in \RR_+^{\ZZ}$. Fixing $(w,z)\in \SS$, let $\kappa(a)$ be the value of $\kappa$ in \eqref{eq:kappa} where $h$ is replaced by $h+a$. Let $\xi(a)=(u')^{-1}(\kappa(a),z)$. Since $\tilde{T}$ is monotonic by part~\eqref{item:Perov_a}, we have $\xi(a)\le \xi(0)$. Therefore,
	\begin{align*}
		&(\tilde{T}(h+a))_z(w)=\kappa(a)=u'(\xi(a),z) \\
		&=\max\{\EE_z \hat{\beta}\hat{R} 
		    (h+a)(\hat{R}(w-\xi(a))+\hat{Y},\hat{Z}),u'(w,z)\} \\
		&=\max\{\EE_z \hat{\beta}\hat{R}
		    h(\hat{R}(w-\xi(a))+\hat{Y},\hat{Z})+\EE_z\hat{\beta}\hat{R}a_{\hat{Z}}, u'(w,z) \}\\
		&\leq \max\{\EE_z \hat{\beta}\hat{R} 
		    h(\hat{R}(w-\xi(a))+\hat{Y},\hat{Z}),u'(w,z)\}
		    +\EE_z\hat{\beta}\hat{R}a_{\hat{Z}} \\
		&\leq \max\{\EE_z \hat{\beta}\hat{R} 
		    h(\hat{R}(w-\xi(0))+\hat{Y},\hat{Z}),u'(w,z)\}
		    +\EE_z\hat{\beta}\hat{R}a_{\hat{Z}} \\
		&=(\tilde{T}h)_z(w)+(Ba)_z.
	\end{align*}
	Therefore, $\tilde{T}(h+a)\le \tilde{T}h+Ba$ pointwise.
\end{proof}

\begin{proof}[Proof of Theorem \ref{t:opt}]
	By Lemma~\ref{lm:Perov} and Theorem 7.6 of \cite{toda2025essential}, there exists a 
	unique fixed point $h^*\in \hH$ of $\tilde{T}$. From the topological 
	conjugacy result in Lemma~\ref{lm:conjug}, we have $\tilde{T} = H T 
	H^{-1}$, so there exists a unique fixed point $c^*\in \cC$ of $T$, and 
	claim~(i) is verified. 
	
	To see that claim~(2) holds, take any $\sigma \in (r(B),1)$. With a slight 
	abuse of notation, let $\|\cdot\|$ denote the supremum norm on $\RR^{\ZZ}$, 
	the operator norm it induces on $\RR^{\ZZ\times \ZZ}$, as well as the 
	supremum norm on the space of functions from $(0,\infty)$ to $\RR^{\ZZ}$. 
	Since $\tilde{T}^n=HT^nH^{-1}$ and $\tilde{T}$ is a Perov contraction with 
	coefficient matrix $B$ with spectral radius $r(B)<1$, it follows from the 
	definition of $\rho$, $\cC$, $\hH$ and Lemma~\ref{lm:Lphi} that
	\begin{align*}
		\rho(T^n c, c^*)
		&=\rho(T^n c, T^n c^*)=\|HT^nc-HT^nc^*\|
		    =\| \tilde{T}^nHc-\tilde{T}^nHc^* \| \\
		&\leq \|B^n\| \|Hc-Hc^*\| 
		    \leq \sigma^n\|Hc-Hc^*\|
	\end{align*}
	for large enough $n$. Therefore, $\rho(T^n c, c^*)\to 0$ as $n\to \infty$.
\end{proof}

Our next goal is to prove Proposition~\ref{pr:monotonea}. We begin by defining
\begin{equation*}
	\cC_0 = \left\{
	c \in \cC \colon 
	w \mapsto w - c(w,z) \text{ is increasing for all } z \in \ZZ 
	\right\}.
\end{equation*}

\begin{lemma}\label{lm:cC'}
	$\cC_0$ is a closed subset of $\cC$, and $Tc \in \cC_0$ for all $c \in \cC_0$. 
\end{lemma}

\begin{proof}
	To see that $\cC_0$ is closed, let $\{c_n\}$ be a sequence in $\cC_0$, and 
	suppose $c \in \cC$ satisfies $\rho(c_n, c) \to 0$. We claim that $c \in 
	\cC_0$. To see this, for each $n$, the map $w \mapsto w - c_n(w,z)$ is 
	increasing for all $z$, and the convergence $\rho(c_n,c) \to 0$ implies 
	pointwise convergence $c_n (w,z) \to c(w,z)$ for all $(w,z) \in \SS$. 
	Thus the monotonicity property is preserved in the limit, so $c \in \cC_0$.
	
	Fix $c \in \cC_0$. We now show that $\xi := Tc \in \cC_0$. By 
	Proposition~\ref{pr:self_map}, $\xi \in \cC$, hence it suffices to show 
	that $w \mapsto w - \xi(w,z)$ is increasing. Suppose not. Then there exist
	$z \in \ZZ$ and $w_1, w_2 \in (0, \infty)$ such that 
	$w_1 < w_2$ and $w_1 - \xi(w_1, z) > w_2 - \xi (w_2, z)$.
	Since $w_1 - \xi (w_1, z) \geq 0$, $w_2 - \xi (w_2, z) \geq 0$ and 
	$\xi(w_1,z) \leq \xi (w_2, z)$ by Proposition~\ref{pr:self_map}, we must 
	have $\xi(w_1, z) < w_1$ and  $\xi(w_1, z) < \xi(w_2, z)$. However, the 
	definition of $T$ combined with the concavity of $u$ then implies 
	that
	\begin{align*}
		u'(\xi (w_1, z), z) 
		&= \EE_z \hat{\beta} \hat{R} 
		    u'(c(\hat{R} \, [w_1 - \xi(w_1,z)] + \hat{Y}, \hat{Z}),
		        \hat Z) \\
		&\leq \EE_z \hat{\beta} \hat{R} 
		    u'(c(\hat{R} \, [w_2 - \xi(w_2,z)] + \hat{Y}, \hat{Z} ), 
		        \hat Z)
		\leq u'(\xi (w_2, z),z),
	\end{align*}
	which gives $\xi(w_1, z) \geq \xi(w_2, z)$, yielding a contradiction. 
	Hence, $w \mapsto w - \xi (w,z)$ is increasing and $T$ is a self-map on 
	$\cC_0$.
\end{proof}

\begin{proof}[Proof of Proposition~\ref{pr:monotonea}]
	Since $T$ maps elements of the closed subset $\cC_0$ into itself by 
	Lemma~\ref{lm:cC'}, Theorem~\ref{t:opt} implies that $c^* \in \cC_0$. 
	Hence, the stated claims hold.
\end{proof}

\begin{proof}[Proof of Proposition~\ref{pr:monotoneY}]
	Let $T_j$ be the time iteration operator for the income process $j$
	as defined in Proposition~\ref{pr:self_map}. It suffices to show $T_1c
	\leq T_2c$ for all $c \in \cC$. To see this, by the monotonicity of $T_j$, 
	we have $T_jc_1 \leq T_jc_2$ whenever $c_1 \leq c_2$. Thus, if $T_1c \leq 
	T_2c$ for all $c \in \cC$, then for any $c_1,c_2\in \cC$ with $c_1\leq 
	c_2$, $T_1c_1 \leq T_1c_2 \leq T_2c_2$. Iterating from any $c\in \cC$ and 
	using Theorem~\ref{t:opt}, we obtain 
	$c_1^* = \lim_{n \to \infty}(T_1)^nc \leq \lim_{n \to \infty}(T_2)^nc 
	= c_2^*$, which proves the claim once $T_1c \leq T_2c$ is established.
	
	To show that $T_1c \leq T_2c$ for any $c\in \cC$, take any $(w,z)\in
	\SS$ and define $\xi_j=(T_jc)(w,z)$. To show $\xi_1 \leq \xi_2$, suppose
	on the contrary that $\xi_1 > \xi_2$. Since $c$ is increasing in $w$, it follows from the definition of the time iteration operator in \eqref{eq:tio}, $Y_1 \leq Y_2$, and $u''(\cdot, z)<0$ that
	\begin{align*}
		u'(\xi_2,z)>u'(\xi_1, z)
		&=\max \{\EE_{z} \hat{\beta} \hat{R}
		u'(c(\hat{R}(w - \xi_1) + \hat{Y}_1, \hat{Z}), \hat Z), 
		u'(w, z) \} \\
		&\geq \max \{ \EE_{z} \, \hat{\beta} \hat{R} 
		u'(c(\hat{R}(w - \xi_2) + \hat{Y}_2, \hat{Z}), \hat Z), 
		u'(w, z) \} 
		= u'(\xi_2, z),
	\end{align*}
	which is a contradiction. This completes the proof.
\end{proof}

\begin{proof}[Proof of Proposition~\ref{pr:binding}]
	Recall that, for all $c \in \cC$, $\xi(w,z) := Tc(w,z)$ solves
	\begin{equation}
		\label{eq:T_opr_general}
		u'(\xi(w,z),z) = \max \left\{ 
		    \EE_{z} \hat{\beta} \hat{R} 
		    u'(c (\hat{R} \, [w - \xi(w,z)] + \hat{Y}, \hat{Z}),
		    \hat Z), u'(w, z) \right\}.
	\end{equation}
	%
	For each $z \in \ZZ$ and $c \in \cC$, define
	\begin{equation}\label{eq:a_bar}
		\bar{w}_c (z) := u'(\cdot, z)^{-1} \left[
		    \EE_z \hat{\beta} \hat{R} u'(c(\hat{Y}, \hat{Z}), \hat Z) 
		\right]
		\quad \text{and} \quad
		\bar{w}(z) := \bar{w}_{c^*} (z).
	\end{equation}
	Let $w \leq \bar{w}_c (z)$. We claim that $\xi(w,z) = w$. Suppose to
	the contrary that $\xi(w,z) < w$. Then 
	$u'(\xi(w,z), z) > u'(w, z)$. In view of 
	\eqref{eq:T_opr_general}, we have
	\begin{equation*}
		u'(w, z) 
		< \EE_z \hat{\beta} \hat{R} 
		    u'(c(\hat{R} \left[w - \xi(w,z) \right] + \hat{Y}, \hat{Z}), 
		    \hat Z) 
		\leq \EE_z \hat{\beta} \hat{R} 
		u'(c(\hat{Y}, \hat{Z}), \hat Z)
		= u'(\bar{w}_c (z), z).
	\end{equation*}
	From this we get $w > \bar{w}_c (z)$, which is a contradiction. Hence,
	$\xi(w,z) = w$.
	
	On the other hand, if $\xi(w,z) = w$, then 
	$u'(\xi(w,z), z) = u'(w,z)$. 
	By \eqref{eq:T_opr_general}, we have
	$u'(w, z) \geq \EE_z \, \hat{\beta} \hat{R}
	u'(c(\hat{Y}, \hat{Z}), \hat Z)
	= u'(\bar{w}_c (z), z)$.                    
	Hence, $w \leq \bar{w}_c (z)$. The first claim is verified. The second 
	claim follows immediately from the first claim and the fact that $c^*$ is 
	the unique fixed point of $T$ in $\cC$.
\end{proof}

\section{Proofs of Section~\ref{ss:ampc=0} Results}

We organize the proofs according to the two mechanisms underlying vanishing MPCs: a central mechanism driven by downward transitions in risk aversion, and a knife-edge mechanism based on spectral radius conditions. 

To establish our key results, we extend the operator $F_i$ defined in 
\eqref{eq:F_oper0} as follows. With slight abuse of notation, in what follows, 
for each $i \in \{1, \dots, N\}$ and $j \in \{1, \dots, M\}$, we define 
$F_i: [1,\infty]^{M} \to [1,\infty]^{M}$ by
\begin{equation}
	\label{eq:F_oper}
	(F_ix)(\tilde z_{j}):=
	\begin{cases}
			\infty, & \text{if\, $\sum_{\ell=1}^{i-1}\bar{p}_{i\ell}>0$}, \\
			\left(1+(G_ix)(\tilde z_{j})^{1/\gamma_{i}}\right)^{\gamma_{i}}, 
			& \text{if\, $\sum_{\ell=1}^{i-1}\bar{p}_{i\ell}=0$}.
		\end{cases}
\end{equation}
Moreover, with slight abuse of notation, we define
\begin{equation*}
	\gamma_i = \gamma(\bar z_i) = \gamma(z_{ij})
	\quad \text{for all } \, z_{ij} = (\bar z_i, \tilde z_j) \in \ZZ. 
\end{equation*}
%

\begin{lemma}
	\label{lm:limsup}
	Fix $i \in \{1, \dots, N\}$. If $c \in \cC$ and for all 
	$j \in \{1, \dots, M\}$, we have
	\begin{equation*}
		\limsup_{w\to \infty} \frac{c(w,z_{ij})}{w} 
		\leq x(\tilde z_{j})^{-1/\gamma_{i}},
	\end{equation*}
	where $x(\tilde z_j) \in [1,\infty]$, then for all $j \in \{1, \dots, M\}$, we have
	\begin{equation}
		\label{eq:lm-limsup}
		\limsup_{w\to \infty} \frac{Tc(w,z_{ij})}{w} 
		\leq (F_ix)(\tilde z_{j})^{-1/\gamma_{i}}.
	\end{equation}
\end{lemma}

\begin{proof}
	Let $\alpha(z_{ij})=\limsup_{w\to \infty} Tc(w,z_{ij})/w$. By definition, we can take an increasing sequence $\{w_{n}\}$ such that $\alpha(z_{ij})=\lim_{n\to \infty} Tc(w_n,z_{ij})/w_n$. Define $\alpha_{n}(z_{ij}):= Tc(w_n,z_{ij})/w_n \in (0,1]$ and 
	\begin{equation}
		\label{eq1:lm-limsup}
		\lambda_{n}(z_{ij})
		=c \big( 
			\hat R (1-\alpha_{n}(z_{ij}))w_{n} + \hat Y, \hat Z 
		\big) \big/ w_{n} >0.
	\end{equation}
	We next show that
	\begin{equation}
		\label{eq:alpha_n_contrdsup}
		\limsup_{n\to \infty}  \lambda_{n}(z_{ij}) \leq x(\hat Z)^{-1/\gamma(\hat Z)}\hat R (1-\alpha(z_{ij})).
	\end{equation}
	Note that 
	\begin{equation*}
		\lambda_{n}(z_{ij}) = 
		\frac{c \big(
			\hat R (1-\alpha_{n}(z_{ij}))w_{n} + \hat Y, \hat Z
			\big)}{\hat R (1-\alpha_{n}(z_{ij}))w_{n} + \hat Y}
		\left(
		\hat R (1-\alpha_{n}(z_{ij})) + \frac{\hat Y}{w_{n}} 
		\right)
	\end{equation*}
	If $\alpha(z_{ij})<1$ and $\hat R>0$, then since $\hat R (1-\alpha_{n}(z_{ij}))w_{n}\to \hat R (1-\alpha(z_{ij}))\cdot \infty = \infty$, by assumption we have
	\begin{equation*}
		\limsup_{n\to \infty} \lambda_{n}(z_{ij}) 
		\leq \limsup_{w\to \infty} 
		\frac{c(w,\hat Z)}{w} \hat R (1-\alpha(z_{ij})) 
		\leq x(\hat Z)^{-1/\gamma(\hat Z)}\hat R (1-\alpha(z_{ij})),
	\end{equation*}
	which is equation (\ref{eq:alpha_n_contrdsup}). If $\alpha(z_{ij})=1$ or $\hat R=0$, since $c(w, z)\leq w$, we have
	\begin{equation*}
		\lambda_{n}(z_{ij})
		\leq \hat R (1-\alpha_{n}(z_{ij})) + \hat Y/w_{n} 
		\to \hat R (1-\alpha(z_{ij}))=0,
	\end{equation*}
	so again equation (\ref{eq:alpha_n_contrdsup}) holds.
	
	Since $Tc(w_{n},z_{ij})=\alpha_{n}(z_{ij})w_{n}$ solves the Euler equation, we have
	\begin{align*}
		\alpha_{n}(z_{ij})^{-\gamma_{i}} 
		&=\max \left\{ 
		\EE_{z_{ij}} \hat \beta \hat R  
		\frac{
			c \big(
			\hat R (w_{n}-\alpha_{n}(z_{ij})w_{n}) + \hat Y, \hat Z 
			\big)^{-\gamma(\hat Z)}
		}{w_n^{-\gamma_{i}}}, \; 1 
		\right\} \\
		&= \max \left\{
		\EE_{z_{ij}} \hat \beta \hat R 
		\lambda_{n}(z_{ij})^{-\gamma(\hat Z)}
		w_{n}^{\gamma_{i}-\gamma(\hat Z)}, \; 1 
		\right\}.
	\end{align*}
	Therefore, 
	\begin{equation}
		\label{eq3:lm-limsup}
		\alpha_{n}(z_{ij})^{-\gamma_{i}} 
		\geq \EE_{z_{ij}} \hat \beta \hat R 
		\lambda_{n}(z_{ij})^{-\gamma(\hat Z)}
		w_{n}^{\gamma_{i}-\gamma(\hat Z)}.
	\end{equation}
	\textit{Case 1.} If $\sum_{\ell=1}^{i-1}\bar{p}_{i\ell}>0$, then there exists at least one $h<i$ such that $\bar p_{ih} > 0$. By definition, this implies that $\gamma_{i}-\gamma_{h} > 0$ with positive probability. Letting $n \to \infty$ in \eqref{eq3:lm-limsup} and applying Fatou's lemma, we have
	\begin{align*}
		\alpha(z_{ij})^{-\gamma_{i}} 
		&= \lim_{n \to \infty} \alpha_{n}(z_{ij})^{-\gamma_{i}} 
		\geq \liminf_{n\to \infty}  
		\EE_{z_{ij}} \hat \beta \hat R 
		\lambda_{n}(z_{ij})^{-\gamma(\hat Z)}w_{n}^{\gamma_{i}-\gamma(\hat Z)} \\
		&\geq \EE_{z_{ij}} \left( 
		\hat \beta \hat R 
		\big[
		\limsup_{n\to \infty} \lambda_{n}(z_{ij}) 
		\big]^{-\gamma(\hat Z)}
		\lim_{n \to \infty} w_{n}^{\gamma_{i}-\gamma(\hat Z)} 
		\right) \\
		&\geq \EE_{z_{ij}} \left( 
		\hat \beta \hat R \big[
		x(\hat Z)^{-1/\gamma(\hat Z)}\hat R (1-\alpha(z_{ij})) 
		\big]^{-\gamma(\hat Z)} 
		\lim_{n \to \infty} w_{n}^{\gamma_{i}-\gamma(\hat Z)} 
		\right).
	\end{align*}
	Note that by definition, there exists some $k \in \{1, \dots, M\}$ such that $\tilde p_{jk} > 0$. For $\hat Z = z_{hk}$, which indicates $\hat R = R(z_{ij}, z_{hk}, \hat \epsilon)$ and $\hat \beta = \beta(z_{ij},z_{hk}, \hat \epsilon)$, we have
	\begin{align*}
		\alpha(z_{ij})^{-\gamma_{i}} 
		&\geq \bar p_{ih} \tilde p_{jk} \EE_{z_{ij},z_{hk}} \left( 
			\hat \beta \hat R \big[
				x(z_{hk})^{-1/\gamma_h}\hat R (1-\alpha(z_{ij})) 
			\big]^{-\gamma_h} 
			\lim_{n \to \infty} w_{n}^{\gamma_{i}-\gamma_h} 
		\right) = \infty,
	\end{align*}
	where the expectation $\EE_{z_{ij},z_{hk}}$ is taken with respect to the innovation $\hat \epsilon$, and the last equality holds because, in this case, $\PP_{(z,\hat z)}(\hat \beta \hat R > 0) > 0$ for all $(z,\hat z)\in \ZZ^2$ by assumption. In addition, since by definition $(F_ix) (\tilde z_j) \equiv \infty$ in this case, we have 
	\begin{equation*}
		\limsup_{w \to \infty}\frac{Tc(w,z_{ij})}{w}
		=\alpha(z_{ij})\leq 0= (F_i x)(\tilde z_{j})^{-1/\gamma_{i}}.
	\end{equation*}
	\textit{Case 2. }If $\sum_{\ell=1}^{i-1}\bar{p}_{i\ell}=0$ and $\bar{p}_{ii}>0$, then, considering only the transition of state $\bar z_i$ to itself next period, \eqref{eq3:lm-limsup} implies that
	\begin{equation*}
		\alpha_{n}(z_{ij})^{-\gamma_{i}} 
		\geq \bar{p}_{ii} \sum_{k=1}^{M} \tilde{p}_{jk} 
		\EE_{z_{ij}, z_{ik}} \left[
		\hat \beta \hat R \lambda_{n}(z_{ij})^{-\gamma_{i}} 
		\right].
	\end{equation*}
	Letting $n \to \infty$ and applying Fatou's lemma, we get
	\begin{align*}
		\alpha(z_{ij})^{-\gamma_{i}}
		&=\lim_{n \to \infty} \alpha_{n}(z_{ij})^{-\gamma_{i}} 
		\geq \liminf_{n \to \infty} 
		\bar{p}_{ii} \sum_{k=1}^{M} \tilde{p}_{jk} 
		\EE_{z_{ij}, z_{ik}} 
		    \hat \beta \hat R \lambda_{n}(z_{ij})^{-\gamma_{i}}  \\
		&\geq \bar{p}_{ii} \sum_{k=1}^{M} \tilde{p}_{jk} 
		\EE_{z_{ij}, z_{ik}} \hat \beta \hat R 
		    \liminf_{n\to\infty} \lambda_{n}(z_{ij})^{-\gamma_{i}} \\
		&= \bar{p}_{ii} \sum_{k=1}^{M} \tilde{p}_{jk} 
		\EE_{z_{ij}, z_{ik}} \hat \beta \hat R 
		    \big(\limsup_{n\to\infty} \lambda_{n} (z_{ij})\big)^{-\gamma_{i}} \\
		&\geq \bar{p}_{ii} \sum_{k=1}^{M} \tilde{p}_{jk} 
		\EE_{z_{ij}, z_{ik}} \hat \beta \hat R 
		    \big[x(z_{ik})^{-1/\gamma_{i}} \hat R (1-\alpha(z_{ij}))\big]^{-\gamma_{i}}.
	\end{align*}
	Solving the above inequality, we have
	\begin{align*}
		\limsup_{w \to \infty}\frac{Tc(w,z_{ij})}{w}
		&=\alpha(z_{ij}) 
		\leq \frac{1}{1+\left( 
			\bar{p}_{ii} \sum_{k=1}^{M} \tilde{p}_{jk} 
			\EE_{z_{ij}, z_{ik}} \hat \beta \hat R^{1-\gamma_{i}}x(z_{ik})
			\right)^{1/\gamma_{i}}} \\
		&= (F_ix)(\tilde z_{j})^{-1/\gamma_{i}}.
	\end{align*}
	\textit{Case 3.}
	If $\sum_{\ell=1}^{i-1}\bar{p}_{i\ell}=0$ and $\bar{p}_{ii}=0$, then, since 
	$Tc(w,z) \leq w$ and $(F_ix) (\tilde z_{j}) \equiv 1$ for all $(w,z)\in \SS$, \eqref{eq:lm-limsup} trivially holds. The proof is now complete.
\end{proof}

%

\begin{lemma}\label{lm:F_conv}
	Let $G$ be a $M \times M$ nonnegative matrix. Define $F: \RR_{+}^{M} \to \RR_{+}^{M}$ by $Fx=\phi(Gx)$, where $\phi(t)=(1+t^{1/\gamma})^{\gamma}$ and $\gamma >0$ is a constant. Then $F$ has a fixed point $x^{*} \in \RR_{+}^{M}$ if and only if $r(G)<1$, and the fixed point is unique. Take any $x_{0} \in \RR_{+}^{M}$ and define the sequence $\{x_{n}\}_{n=1}^{\infty} \subset \RR_{+}^{M}$ by $x_{n}=Fx_{n-1}$ for all $n \in \mathbb{N}$. Then we have the following:
	\begin{enumerate}
		\item If $r(G) < 1$, then  $\{x_{n}\}_{n=1}^{\infty}$ converges to $x^{*}$.
		\item If $r(G) \geq 1$ and $G$ is irreducible, then $\{x_{n}\}_{n=1}^{\infty}$ converges to $x^{*}=\infty$.
	\end{enumerate}
\end{lemma}

\begin{proof}
	This follows from Lemmas~15--16 of \cite{ma2021theory}.
\end{proof}

\begin{proposition}
	\label{pr:F_fpt}
	For given $i\in \{1,\dots, N\}$, if $\sum_{\ell=1}^{i-1}\bar p_{i\ell} =0$ and $\bar p_{ii}>0$, then $F_i$ defined in \eqref{eq:F_oper} has a fixed point $x_i^* \in \RR_+^{M}$ if and only if $r(G_{i}) < 1$, and the fixed point is unique. Take any $x_{0} \in \RR_+^{M}$, and define the sequence $\{x_{n}\}_{n=1}^{\infty} \subset \RR_{+}^{M}$ by $x_{n}=F_i x_{n-1}$ for all $n \in \mathbb{N}$. Then we have the following:
	\begin{enumerate}
		\item If $r(G_{i}) < 1$, then  $\{x_{n}\}_{n=1}^{\infty}$ converges to $x_i^{*}$.
		\item If $r(G_{i}) \geq 1$ and $G_{i}$ is irreducible, then $\{x_{n}\}_{n=1}^{\infty}$ converges to $x_i^{*}=\infty$.
	\end{enumerate}
\end{proposition}

\begin{proof}
	This follows immediately from Lemma~\ref{lm:F_conv} by letting $F = F_i$ and $G = G_i$. 
\end{proof}

\begin{proof}[Proof of Theorem \ref{t:0ampc_thm}]
	
	Define the sequence $\{c_{n}\} \subset \cC$ by $c_{0}(w,z)=w$ and $c_{n}(w,z)=Tc_{n-1}(w,z)$ for all $n\geq 1$. Since $Tc(w,z) \leq w$ for any $c \in \cC$, we have $c_{1}(w,z)=Tc_{0}(w,z)\leq w=c_{0}(w,z)$. Since $T:\cC \to \cC$ is monotone, by induction $0 \leq c_{n}(w,z) \leq c_{n-1}(w,z) \leq w$ for all $n\geq 1$ and $(w,z)\in \SS$, and $c^{*}(w,z)=\lim_{n \to \infty} c_{n}(w,z)$ exists. This $c^{*}$ is the unique fixed point of $T$ and also the unique optimal policy in $\cC$.
	
	Fix $i \in\{1, \dots, N\}$. Define the sequence $\{x_{n}\} \subset \RR_{+}^{M}$ by 
	\begin{equation*}
		x_{0}(\tilde z_j)=1
		\quad \text{and} \quad 
		x_{n}(\tilde z_{j})=(F_i x_{n-1})(\tilde z_{j}),
		\qquad j = 1, \dots, M.
	\end{equation*}
	By definition, $c_{0}(w,z_{ij})/w=1=x_{0}(\tilde z_{j})^{-1/\gamma_{i}}$. Hence, for all $j \in \{1, \dots,M\}$,
	\begin{equation*}
		\limsup_{w \to \infty}\frac{c_{0}(w,z_{ij})}{w} 
		\leq x_{0}(\tilde z_{j})^{-1/\gamma_{i}}.
	\end{equation*}
	Since $c_{n}(w,z_{ij}) \downarrow c^{*}(w,z_{ij})$ pointwise, a repeated application of Lemma \ref{lm:limsup} yields
	\begin{equation}\label{eq:ampc_bds}
		0 \leq \liminf_{w \to \infty}\frac{c^{*}(w,z_{ij})}{w} 
		\leq \limsup_{w \to \infty}\frac{c^{*}(w,z_{ij})}{w} 
		\leq \limsup_{w \to \infty}\frac{c_{n}(w,z_{ij})}{w} 
		\leq x_{n}(\tilde z_{j})^{-1/\gamma_{i}}.
	\end{equation}
	\textit{Case 1.} If $\sum_{j=1}^{i-1}\bar{p}_{ij}>0$, then Lemma~\ref{lm:limsup} implies that $x_{n}(\tilde z_{j})=\infty$ for all $n\geq 1$ and $j\in \{1,\dots,M\}$. Then \eqref{eq:0ampc_thm} follows from \eqref{eq:ampc_bds}.
	
	\textit{Case 2.} If $\bar{p}_{ii}>0$, $G_{i}$ is irreducible, and $r(G_i) \geq 1$, then by Proposition~\ref{pr:F_fpt}, we have $x_{n}(\tilde z_{j}) \to \infty$ as $n \to \infty$ for all $j\in \{1,\dots,M\}$. Again, \eqref{eq:0ampc_thm} follows from \eqref{eq:ampc_bds}. 
\end{proof}

\begin{proof}[Proof of Theorem~\ref{t:0ampc_thm_pos}]
	If each entry of $\bar P$ is positive, then condition~(1) of 
	Theorem~\ref{t:0ampc_thm} holds. If in addition $G_1$ is irreducible and 
	$r(G_1) \geq 1$, then condition~(2) of Theorem~\ref{t:0ampc_thm} holds.
	In both cases, so the conclusion follows trivially. 
\end{proof}

\section{Proofs of Section~\ref{ss:ampc>0} Results}

We now prove the results of Section~\ref{ss:ampc>0}, where downward transitions in risk aversion occur with zero probability (i.e., $\sum_{j=1}^{i-1}\bar{p}_{ij}=0$). For expositional clarity, we begin with the polar case $\bar p_{ii} = 1$ (Theorem~\ref{t:non0_ampc1}), where risk aversion remains fixed along relevant histories, before addressing the more intricate cases $\bar p_{ii} = 0$ (Theorem~\ref{t:non0_ampc2} and Proposition~\ref{pr:non-concave}) and $\bar p_{ii} \in (0,1)$ (Theorem~\ref{t:non0_ampc3}). These latter cases generate fundamentally new asymptotic consumption patterns, and require distinct proof strategies with significant technical challenges. The proofs are presented in this order to reflect methodological dependencies, rather than the relative significance of the results.

To prove Theorem \ref{t:non0_ampc1}, we need the following lemma.

\begin{lemma}
	\label{lm:liminfN}
	Fix $i\in \{1, \dots, N\}$. Consider $F_i$ defined in \eqref{eq:F_oper} with  $\bar{p}_{ii}=1$. Suppose $r(G_{i}) < 1$ and let $x_{i}^{*} \in \RR_{+}^{M}$ be the unique fixed point of $F_i$. If $c \in \cC$ and for all $w>0$ and $j \in \{1, \dots, M\}$, we have 
	\begin{equation*}
		\frac{c(w,z_{ij})}{w} \geq x_i^{*}(\tilde z_{j})^{-1/\gamma_i},
	\end{equation*}
	then for all $w>0$ and $j \in \{1, \dots, M\}$, we have 
	\begin{equation*}
		\frac{Tc(w,z_{ij})}{w} \geq x_i^{*}(\tilde z_{j})^{-1/\gamma_{i}}.
	\end{equation*}
\end{lemma}

\begin{proof}
	Suppose $Tc(w,z_{ij}) / w < x_i^{*}(\tilde z_{j})^{-1/\gamma_{i}}$ for some 
	$w$ and $z_{ij}$. Let $\xi := Tc(w,z_{ij})$. Then 
	$x_i^{*}(\tilde z_{j})^{-1/\gamma_{i}} w >\xi$. Since in addition 
	$x_i^{*}(\tilde z_{j})^{-1/\gamma_{i}} \in (0,1]$ by 
	Proposition~\ref{pr:F_fpt} (claim~(1)), by the monotonicity of $u'$, we have
	\begin{align*}
		w^{-\gamma_{i}}
		&= u'(w,z_{ij})
		\leq u'(x_i^{*}(\tilde z_{j})^{-1/\gamma_{i}} w,z_{ij}) 
		< u'(\xi,z_{ij}) \\
		&= \max \left\{
		\EE_{z_{ij}} \hat \beta \hat R  
		u' (c(\hat R (w-\xi) + \hat Y, \hat Z), \hat Z),
		\; u'(w,z_{ij})
		\right\} \\
		&=\max \left\{
		\EE_{z_{ij}} \hat \beta \hat R  
		c (\hat R (w-\xi) + \hat Y, \hat Z)^{-\gamma_{i}},
		\; w^{-\gamma_{i}}\right\}.
	\end{align*}
	Recall that $\hat Z = (\hat{\bar Z}, \hat{\tilde Z})$ by definition.
	Therefore, the third equality above follows from $\bar p_{ii}=1$, which 
	implies $\hat{\bar Z} = \bar z_i$. Consequently, 
	$\hat \gamma = \gamma(\bar z_i) = \gamma_i$ with probability one 
	conditional on $Z = z_{ij}$. We then have
	\begin{equation*}
		w^{-\gamma_{i}} < \EE_{z_{ij}} \hat \beta \hat R  
		c (\hat R (w-\xi) + \hat Y, \hat Z)^{-\gamma_{i}}.
	\end{equation*}
	Under the contradiction hypothesis, we have
	\begin{align*}
		x_i^{*}(\tilde z_{j})w^{-\gamma_{i}}
		&< \xi^{-\gamma_i} 
		=\EE_{z_{ij}} \hat \beta \hat R  
		c (\hat R (w-\xi) + \hat Y, \hat Z)^{-\gamma_{i}} \\
		&\leq \EE_{z_{ij}} \hat \beta \hat R  
		x_i^{*}(\hat{\tilde{Z}}) (\hat R (w-\xi) + \hat Y)^{-\gamma_{i}} \\
		&< \EE_{z_{ij}} \hat \beta \hat R^{1-\gamma_{i}}  
		x_i^{*}(\hat{\tilde Z}) \left(
			1-x_i^{*}(\tilde z_{j})^{-1/\gamma_{i}} 
		\right)^{-\gamma_{i}} w^{-\gamma_{i}}.
	\end{align*}
	By simple algebra and using the fact that $\bar p_{ii} = 1$, we obtain 
	\begin{align*}
		x_i^{*}(\tilde z_{j}) &< \left(
			1 + \left[\EE_{z_{ij}} \hat \beta \hat R^{1-\gamma_i}
			x_i^*(\hat{\tilde Z}) \right]^{1/\gamma_i}
		\right)^{\gamma_i}  \\
		&= \left(
			1 + \left[\bar p_{ii} \sum_{k=1}^M \tilde p_{jk} 
			\EE_{z_{ij},z_{ik}} \hat\beta \hat R^{1-\gamma_i}
			x_i^*(\tilde z_k) \right]^{1/\gamma_i}
		\right)^{\gamma_i}
		= (F_ix_i^*)(\tilde z_j).
	\end{align*}
	This contradicts the definition of $F_i$ and $x_i^{*}(\tilde z_j)$.
\end{proof}

\begin{proof}[Proof of Theorem~\ref{t:non0_ampc1}]
	Same to the proof of Theorem~\ref{t:0ampc_thm}, we define the sequence $\{c_{n}\} \subset \cC$ by $c_{0}(w,z)=w$ and $c_{n}=Tc_{n-1}$ for all $n\geq 1$, which gives a decreasing sequence of functions $0 \leq c_{n} \leq c_{n-1} \leq w$ for all $n$, and the optimal policy $c^{*}(w,z)=\lim_{n \to \infty} c_{n}(w,z)$ is well defined.
	Moreover, define the sequence $\{x_{n}\} \subset \RR_{+}^{M}$ by $x_{0}(\tilde z_j)=1$ and $x_{n}(\tilde z_j)=(F_i x_{n-1})(\tilde z_j)$ for all $j \in\{1,\dots,M\}$. A repeated application of Lemma \ref{lm:limsup} gives
	\begin{equation}\label{eq1:prop-limN}
		\limsup_{w \to \infty}\frac{c^{*}(w,z_{ij})}{w} 
		\leq x_i^{*}(\tilde z_{j})^{-1/\gamma_{i}}.
	\end{equation}
	Since $c_0(w,z_{ij}) / w = 1 \geq x_i^*(\tilde z_j)^{-1/\gamma_i}$ for all $w>0$ and $z_{ij}$, a repeated application of Lemma \ref{lm:liminfN} implies that $c_{n}(w,z_{ij}) /w \geq x_i^{*}(\tilde z_{j})^{-1/\gamma_{i}}$ for all $w>0$ and $z_{ij}$. Since $c_{n}\to c^{*}$ pointwise, letting $n\to \infty$, dividing both sides by $w>0$, and letting $w \to \infty$, we obtain
	\begin{equation}\label{eq2:prop-limN}
		\liminf_{w \to \infty}\frac{c^{*}(w,z_{ij})}{w} 
		\geq x^{*}(\tilde z_j)^{-1/\gamma_{i}}
	\end{equation}
	for all $j \in \{1, \dots, M\}$. Then the stated claim follows from
	\eqref{eq1:prop-limN} and \eqref{eq2:prop-limN}.
\end{proof}

To prove Theorem \ref{t:non0_ampc2} and Theorem~\ref{t:non0_ampc3}, we need the following lemma.

\begin{lemma}\label{lm:lim}
	Fix $i \in \{1,\dots, N\}$. Suppose the following conditions hold:
	\begin{enumerate}
		\item $\sum_{j=1}^{i-1}\bar{p}_{ij}=0$.
		\item There exists $m>0$ such that $R \geq m$ with probability one.
	\end{enumerate}
	If $c \in \cC$ and for all $z_{ij} \in \ZZ$, we have
	\begin{equation*}
		\lim_{w \to \infty}\frac{c(w,z_{ij})}{w} = x(\tilde{z}_{j})^{-1/\gamma_{i}},
	\end{equation*}
	then for all $z_{ij} \in \ZZ$, we have
	\begin{equation}
		\label{eq:lm-lim}
		\lim_{w \to \infty}\frac{Tc(w,z_{ij})}{w} = (F_{i}x)(\tilde{z}_{j})^{-1/\gamma_{i}}.
	\end{equation}
\end{lemma}

\begin{proof}
	Fix $i \in \{1, \dots, N\}$. If $\PP_{z_{ij}} (\hat \beta \hat R > 0) = 0$ 
	for some $j$, then by the definition of $T$ and $G_i$, we have 
	$Tc(w,z_{ij}) \equiv w$ and $(G_i x)(\tilde z_j) \equiv 0$. In this case, 
	\eqref{eq:lm-lim} holds trivially since $\lim_{w\to\infty} Tc(w,z_{ij})/w = 
	1 = (F_ix)(\tilde z_j)^{-1/\gamma_i}$. In what follows, we consider 
	$\PP_{z_{ij}} (\hat \beta \hat R > 0) > 0$ for all $j$. 
	
	In this case, if $x(\tilde z_j) = \infty$ for some $j$, then 
	\eqref{eq:lm-lim} holds trivially by Lemma~\ref{lm:limsup}. In what 
	follows, we consider $x(\tilde z_j) < \infty$ for all $j$. 
	
	Let $\alpha(z_{ij})$ be an accumulation point of $Tc(w,z_{ij})/w$ as $w \to \infty$. By Lemma~\ref{lm:limsup}, we have $\alpha(z_{ij}) \leq 1$. By definition, we can take an increasing sequence $\{w_{n}\}$ such that $\alpha(z_{ij})=	\lim_{n \to \infty} Tc(w_{n},z_{ij})/w_{n}$. Define $\alpha_{n}(z_{ij})=Tc(w_{n},z_{ij})/w_{n} \in (0,1]$ and $\lambda_{n}(z_{ij})$ as in Lemma \ref{lm:limsup}. By the proof of Lemma \ref{lm:limsup}, we have
	\begin{equation}
		\label{eq1:lm-lim}
		\lim_{n \to \infty}  \lambda_{n}(z_{ij}) = x(\hat Z)^{-1/\gamma(\hat Z)}\hat R (1-\alpha(z_{ij})).
	\end{equation}
	Since $Tc(w_{n},z_{ij})=\alpha_{n}w_{n}$ solves the Euler equation, we have
	\begin{align}
		\nonumber
		\alpha_{n}(z_{ij})^{-\gamma_{i}} 
		&= \max \left\{
		\EE_{z_{ij}} \hat \beta \hat R  
		\frac{c (\hat R (w_{n}-\alpha_{n}(z_{ij})w_{n}) + \hat Y, \hat Z)^{-\gamma(\hat Z)}}
		{w^{-\gamma_{i}}},
		\; 1 
		\right\} \\ \label{eq:alpha_n_contrd}
		&= \max \left\{
		\EE_{z_{ij}} \hat \beta \hat R 
		\bigg[\frac{c(\hat R (w_{n}-\alpha_{n}(z_{ij})w_{n}) + \hat Y, \hat Z)}
		{w_{n}}\bigg]^{-\gamma(\hat Z)}w_{n}^{\gamma_{i}-\gamma(\hat Z)},
		\; 1 
		\right\}.
	\end{align}
	Since $\hat R \geq m >0$ with probability one by assumption,
	\eqref{eq1:lm-lim} implies that
	\begin{equation}
		\label{lim-equation1}
		\liminf_{n \to \infty}  \lambda_{n}(z_{ij}) \geq  x(\hat Z)^{-1/\gamma(\hat Z)}m (1-\alpha(z_{ij})).
	\end{equation}
	\textit{Case 1.} $\bar p_{ii} = 0$.
	If $\bar p_{ii} = 0$, then $(F_i x)(\tilde z_j) \equiv 1$. It suffices to show that $\alpha (z_{ij}) = 1$ for all $j$. Because in this case, we have $\alpha(z_{ij}) = (F_i x)(\tilde z_j)$, then the stated claim holds. 
	
	Suppose $\alpha(z_{ij}) < 1$ for some $j$, then $\lim_{n\to\infty} \alpha_n(z_{ij}) < 1$, implying that $\alpha_n(z_{ij}) < 1$ for all $n\geq M$ for some integer $M$. By \eqref{eq:alpha_n_contrd}, we have
	\begin{equation}\label{eq:alpha_n}
		\alpha_{n}(z_{ij})^{-\gamma_{i}} 
		= \EE_{z_{ij}} \hat \beta \hat R 
		\bigg[\frac{c(\hat R (w_{n}-\alpha_{n}(z_{ij})w_{n}) + \hat Y, \hat Z)}
		{w_{n}}\bigg]^{-\gamma(\hat Z)}w_{n}^{\gamma_{i}-\gamma(\hat Z)}.
	\end{equation}
	Since we now work with $x(\tilde z_j) < \infty$ for all $j$, as pointed out above. By the definition of $F_i$, we have $x(\tilde z_{j}) \in [1,\infty)$ for all $j$. Hence, for any
	\begin{equation*}
		\underline{\lambda} \in \left(0, \min_{j}x(\tilde z_{j})^{-1/\gamma_{i}}m (1-\alpha(z_{ij})) \right)
		\quad \text{and} \quad 
		\bar{\alpha} \in \left(\max_j\alpha(z_{ij}),1 \right),
	\end{equation*}
	there exists $N \in \mathbb{N}$ such that $\lambda_{n} \geq \underline{\lambda}$, $\alpha_{n} < \bar{\alpha}$, and $w_{n} \geq 1$ for all $n \geq N$ and $z_{ij}$. The integrand in the expectation in (\ref{eq:alpha_n_contrd}) is bounded above by $\max_{z, \hat z\in \ZZ} \hat \beta \hat R \underline{\lambda}^{-\gamma(\hat z)} 
	< \infty$,
	which is integrable. Letting $n \to \infty$, applying the dominated convergence theorem, and using $\alpha(z_{ij}) <1$, we obtain
	\begin{equation*}
		\label{eq3:lm-lim}
		\alpha(z_{ij})^{-\gamma_{i}}
		= \lim_{n \to \infty} \EE_{z_{ij}} \hat \beta \hat R \lambda_n(z_{ij})^{-\gamma(\hat Z)}
		w_{n}^{\gamma_{i}-\gamma(\hat Z)} 
		= \EE_{z_{ij}} \hat \beta \hat R 
		\left[\lim_{n \to \infty} \lambda_n(z_{ij})\right]^{-\gamma(\hat Z)}
		\lim_{n \to \infty} w_{n}^{\gamma_{i}-\gamma(\hat Z)}. \nonumber
	\end{equation*}
	Since $\sum_{\ell=1}^i \bar p_{i\ell} = 0$ implies $\sum_{h=i+1}^N \bar p_{ih} = 1$, and because $\PP_{z_{ij}}(\hat \beta \hat R > 0) > 0$ in this case, and because
	\begin{equation*}
		\left[\lim_{n \to \infty} \lambda_n(z_{ij})
		\right]^{-\gamma(\hat Z)}
		\leq \max_{\hat Z \in \ZZ} 
		\underline{\lambda}^{-\gamma(\hat Z)} < \infty,
	\end{equation*}
	together with the fact that 
	$\lim_{n\to\infty} w_n^{\gamma_i - \gamma_h} = 0$ for all $h > i$,	we 
	obtain the conclusion that $\alpha(z_{ij})^{-\gamma_i} = 0$. This gives 
	$\alpha(z_{ij}) = \infty$, which is contradicted with the assumption that 
	$\alpha(z_{ij}) < 1$. Hence, we must have $\alpha(z_{ij}) = 1$ for all $j$.
	  
    \textit{Case 2.} $\bar p_{ii} > 0$. If $\bar p_{ii} > 0$, then by 
    Lemma~\ref{lm:limsup} we obtain $\alpha(z_{ij}) < 1$ for all $j$.
    Indeed, in this case $\PP_{z_{ij}}(\hat \beta \hat R > 0) > 0$ for all $j$, 
    which guarantees that each row of $G_i$ has at least one positive entry.
    Hence $(G_i x)(\tilde z_j) > 0$ for all $j$, and consequently 
    $(F_i x)(\tilde z_j) > 1$ for all $j$. Hence \eqref{eq:alpha_n} holds. 
	This implies that
	\begin{align*}
		\alpha(z_{ij})^{-\gamma_{i}}
		&=  \bar{p}_{ii} \sum_{k=1}^{M} \tilde{p}_{jk} 
		\EE_{z_{ij}, z_{ik}} 
		    \hat \beta \hat R \left[ \lim_{n \to \infty} \lambda_{n}(z_{ij})\right]^{-\gamma_{i}} \\
		& + \sum_{h=i+1}^{N} \bar{p}_{ih} \sum_{k=1}^{M} \tilde{p}_{jk} 
		\EE_{z_{ij}, z_{hk}} \left( 
		    \hat \beta \hat R \left[\lim_{n \to \infty}\lambda_{n}(z_{ij})\right]^{-\gamma_{h}} \lim_{n \to \infty} w_{n}^{\gamma_{i}-\gamma_{h}} 
		\right)\\
		&= \bar{p}_{ii} \sum_{k=1}^{M} \tilde{p}_{jk} 
		\EE_{z_{ij}, z_{ik}}
  		    \hat \beta \hat R \left(x(\tilde{z}_{k})^{-1/\gamma_{i}} \hat R (1-\alpha(z_{ij})) \right)^{-\gamma_{i}},
	\end{align*}
	where the second equality follows by the same reasoning as in case 1.
	Solving the above equality, we obtain
	\begin{align*}
		\lim_{w \to \infty} \frac{Tc(w,z_{ij})}{w}
		&=\alpha(z_{ij})
		= \frac{1}{1+\left( 
			\bar{p}_{ii} \sum_{k=1}^{M} \tilde{p}_{jk} 
			\EE_{z_{ij}, z_{ik}} 
			\hat \beta \hat R^{1-\gamma_{i}}x(\tilde z_{k})
		\right)^{1/\gamma_{i}}}\\
		&= \frac{1}{1+(G_ix)(\tilde z_j)^{1/\gamma_{i}}}
		= (F_{i}x)(\tilde{z}_{j})^{-1/\gamma_{i}}.
	\end{align*}
	The proof is now complete.
\end{proof}

\begin{proof}[Proof of Theorem~\ref{t:non0_ampc2}]
	Same as before, let $c_0(w,z) \equiv w$ and $c_n = T^nc_0$. A repeated application of Lemma~\ref{lm:limsup} implies that, for all $z_{ij}\in\ZZ$, we have
	\begin{equation*}
	    \limsup_{w\to \infty} \frac{c^*(w,z_{ij})}{w} \leq (F_i^n x)(\tilde z_j).	
	\end{equation*}	
	A repeated application of Lemma~\ref{lm:lim} implies that
	\begin{equation*}
		\liminf_{w\to \infty} \frac{c^*(w,z_{ij})}{w} = (F_i^n x)(\tilde z_j).
	\end{equation*}
	Since in addition $\bar p_{ii} = 0$, we have $F_i^n(\tilde z_j) \equiv 1$ for all $n\geq 1$. The stated claim follows immediately from the above inequalities.	
\end{proof}

\begin{proof}[Proof of Proposition~\ref{pr:non-concave}]
	Fix $z_{ij} \in \ZZ$ and denote $z = z_{ij}$.
	By Proposition~\ref{pr:binding}, we know that $c^*(w,z) =w$ if and only if 
	$w \leq \bar w(z)$. In particular, $c^*(w,z) < w$ for all $w > \bar w(z)$.
	Moreover, since $c^* \in \cC$, there exists some $M<\infty$ such that 
	\begin{equation*}
		\bar w(z) \leq u'(\cdot, z)^{-1} \left[ 
		\EE_z \hat \beta\hat R 
		\left(u'(\hat Y,\hat Z) + M \right) 
		\right] < \infty,
	\end{equation*}
	where the first inequality is due to the definition of $\bar w(z)$ and the 
	monotonicity of $u'(\cdot, z)$, and the second inequality follows from 
	Assumption~\ref{a:spec_and_finite_exp}. Hence, there exist some 
	$w \in (\bar w(z), \infty)$ and $\epsilon > 0$ such that
	\begin{equation*}
		c^*(w,z) < w (1-\epsilon)
		\quad \Longrightarrow \quad 
		c^*(w,z)/w < 1- \epsilon.
	\end{equation*} 
	However, by Theorem~\ref{t:non0_ampc2}, there exists some sufficiently 
	large $w'>w$ such that
	\begin{equation*}
		c^*(w',z)/w' \geq 1- \epsilon.
	\end{equation*}
	As a result, we have
	\begin{equation*}
		c^*(w,z)/w  < c^*(w',z)/w'
	\end{equation*}
	for some $w'>w$. Hence $c^*(w,z)$ is not concave in $w$. The proof is 
	complete.
\end{proof}

To prove Theorem \ref{t:non0_ampc3}, we extend the utility function to 
\begin{equation*}
	u(c,z_{ij}) = \begin{cases}
		\psi(\bar z_i) \frac{c^{1-\gamma(\bar z_i)}}{1-\gamma(\bar z_i)}, & 
		\text{if } \gamma(\bar z_i)>0 \text{ and } \gamma(\bar z_i)\neq 1,  \\
		\psi(\bar z_i) \log c, & \text{if } \gamma(\bar z_i) = 1,
	\end{cases}
\end{equation*}
where $\psi_i := \psi(\bar z_i) > 0$ denotes the weight on $\bar z_i$. We denote 
\begin{equation*}
	\psi := (\psi_1, \dots, \psi_N) = (\psi(\bar z_i), \dots, \psi (\bar z_N)).
\end{equation*}
Furthermore, we introduce a new operator. Fix $i \in \{i, \dots, N\}$, we define the $M\times M$ matrix $D_i$ with entries
\begin{equation*}
	d_{jk} = \tilde p_{jk} \sum_{h=1}^N \bar p_{ih} \psi_h 
	\EE_{z_{ij},z_{hk}} \beta(z_{ij}, z_{hk}, \hat \epsilon) 
	R(z_{ij}, z_{hk}, \hat \epsilon)^{1-\gamma_i},
\end{equation*}
where, as before, the expectation is taken with respect to $\hat \epsilon$. We 
then set
\begin{equation*}
	U_i = \tilde P \circ D_i
\end{equation*}
where $\tilde P \circ D_i$ as the Hadamard product of $\tilde P$ and $D_i$.
For $y \in \RR_+^M$, we define 
\begin{equation}\label{eq:J_ope}
	(J_iy)(\tilde z_j) := \left(
	1 + (U_i y)(\tilde z_j)^{1/\gamma_i}
	\right)^{\gamma_i}.
\end{equation}
%
Denote by $c^*(w,z; \psi, \gamma)$ the optimal consumption at $(w,z)$ when the weight vector is $\psi$ and $\gamma_i \equiv \gamma$ for all $i$. As usual, denote by $c^*(w,z)$ the optimal consumption at $(w,z)$ when $\psi_i \equiv 1$ for all $i$.  

\begin{lemma}
	\label{lm:Clim}
	Fix $i\in \{1, \dots, N\}$. Suppose $\sum_{\ell=1}^{i-1} \bar{p}_{i\ell}=0$ and $\gamma \equiv \gamma_i$. If $r(U_i) < 1$, then $J_i: \RR_+^M \to \RR_+^M$ and $J_i$ has a unique fixed point $y_i^* \in \RR_+^M$. Moreover, for all $j \in \{1, \dots, M\}$, we have 
	\begin{equation}
		\label{eq:lm-Clim}
		\lim_{w \to \infty} \frac{c^*(w,z_{ij}; \psi, \gamma_i)}{w}
		=y_i^*(\tilde z_{j})^{-1/\gamma_i}.
	\end{equation}
\end{lemma}

\begin{proof}
	The first claim follows immediately from Lemma~\ref{lm:F_conv}, which implies that $J_i$ has a unique fixed point in $y_i^* \in \RR_+^M$. The proof of the second claim is similar to the proof of Theorem~\ref{t:non0_ampc1} and thus omitted. 
\end{proof}

%

Let $T_{0}$ be the time iteration operator for the special case $\gamma_i\equiv \gamma$.

\begin{proposition}[Parametric monotonicity with respect to risk attitudes]
	\label{pr:pm}
	Fix $i \in \{1, \dots, N\}$. Suppose $\sum_{\ell=1}^{i-1} \bar{p}_{i\ell}=0$ and there exists $b>0$ such that $Y(z, \hat z, \hat \epsilon)\geq b$ almost surely conditional on $(z,\hat z)$. Choose $\psi$ such that $b^{-\gamma_{h}} / \psi_h \leq b^{-\gamma_{i}}/\psi_{i}$
	for all $h \geq i$. Then $c^*(w, z_{ij}; \psi, \gamma_i) \leq c^*(w,z_{ij})$ for all $w \geq b$ and $z_{ij} \in \ZZ$.
\end{proposition}

\begin{proof}
	Consider the initial candidate $c(w,z) \equiv w$. Since
	\begin{equation*}
		(T_0^{n}c)(w,z_{ij}) \to c^*(w,z_{ij}; \psi, \gamma_{i})
		\quad \text{and} \quad 
		(T^{n}c)(w,z_{ij}) \to c^{*}(w,z_{ij}),
	\end{equation*}
	for all $(w,z_{ij})$, to prove the stated claim, it suffices to show that
	\begin{equation*}
		(T_0^{n}c)(w,z_{ij}) \leq (T^{n}c)(w,z_{ij})
		\quad \text{for all $n \in \NN$, $w\geq b$, and $z_{ij}\in \ZZ$}.
	\end{equation*}
	Denote $c_n^0 := T_0^n c$ and $c_n := T^n c$. The claim obviously holds for $n=0$. Suppose it holds for arbitrary $n$. Then $c_n^0(w,z_{ij}) \leq c_n(w,z_{ij})$. It remains to verify that $c_{n+1}^0(w,z_{ij}) \leq c_{n+1}(w,z_{ij})$ for all $w\geq b$ and $z_{ij}\in \ZZ$. Suppose this is not true. Then there exists $w\geq b$ and $z_{ij} \in \ZZ$ such that $c_{n+1}^0(w,z_{ij}) > c_{n+1}(w,z_{ij})$. Then we have $1 \geq c_{n+1}^0(w,z_{ij})/w > c_{n+1}(w,z_{ij})/w.$
	By the definition of $T$ and the induction argument, we have
	\begin{align}
		c_{n+1}(w,z_{ij})^{-\gamma_{i}}
		&= \EE_{z_{ij}} \hat \beta \hat R  c_n (\hat R (w-c_{n+1}(w,z_{ij})) + \hat Y, \hat Z)^{-\hat \gamma} \nonumber \\
		&\leq \EE_{z_{ij}} \hat \beta \hat R  c_n (\hat R (w-c_{n+1}^0(w,z_{ij})) + \hat Y, \hat Z)^{-\hat \gamma} \nonumber \\
		&\leq \EE_{z_{ij}} \hat \beta \hat R  c_n^0 (\hat R (w-c_{n+1}^0(w,z_{ij})) + \hat Y, \hat Z)^{-\hat \gamma}. \label{eq:cnc0}
	\end{align}
	To proceed, we show that as long as $m$ is chosen sufficiently small, we have 
	\begin{equation}\label{eq:cn_lb}
		c_n^0(w,z) \geq b 
		\quad \text{for all } n \geq 0, w\geq b, \text{ and } z\in\ZZ. 
	\end{equation}
	To see this, recall that based on the monotonicity of $T$, we have 
	$c_{n}^0 \leq c_{n-1}^0$ for all $n \geq 1$. Recall $\bar w_c$ and $\bar w$ defined in \eqref{eq:a_bar}. To simplify notation, we denote $\bar w_n := \bar w_{c_n^0}$. Then based on the monotonicity of $u$, we have $\bar w_n \geq \bar w$ for all $n \geq 0$. Moreover, by definition, $\bar w(z) > 0$ for all $z\in \ZZ$. Hence, we can choose $b$ such that $b \in \left(0, \min_{z\in \ZZ}\bar{w}(z)\right)$.
	Then we have $c_n^0 (b, z) = b$ by Proposition~\ref{pr:binding}. Then, by the monotonicity of $c_n^0$, we have $c_n^0(w,z) \geq c_n^0(b,z) = b > 0$ for all $w\geq b, z\in \ZZ$, and $n \geq 0$. This combined with \eqref{eq:cnc0} implies that
	\begin{align*}
		c_{n+1}(w,z_{ij})^{-\gamma_i}
		& \leq \EE_{z_{ij}} \hat \beta \hat R b^{-\hat \gamma}
		\left[
		    \frac{c_n^0 (\hat R (w-c_{n+1}^0(w,z_{ij})) + \hat Y, \hat Z)}{b}
		\right]^{-\hat \gamma}  \\
		&= \bar{p}_{ii} \sum_{k=1}^{M} \tilde{p}_{jk} 
		\EE_{z_{ij}, z_{ik}}
		\hat \beta \hat R b^{-\gamma_{i}} \left[\frac{c_n^{0} (\hat R (w-c_{n+1}^{0}(w,z_{ij})) + \hat Y, z_{ik})}{b}\right]^{-\gamma_{i}} \\
		&+ \sum_{h=i+1}^{N} \bar{p}_{ih} \sum_{k=1}^{M} \tilde{p}_{jk} 
		\EE_{z_{ij}, z_{hk}}
		\hat \beta \hat R b^{-\gamma_{h}} \left[\frac{c_n^{0} (\hat R (w-c_{n+1}^{0}(w,z_{ij})) + \hat Y, z_{hk})}{b}\right]^{-\gamma_{h}} \\
		&\leq \bar{p}_{ii} \sum_{k=1}^{M} \tilde{p}_{jk} 
		\EE_{z_{ij}, z_{ik}}
		\hat \beta \hat R b^{-\gamma_{i}} \left[\frac{c_n^{0} (\hat R (w-c_{n+1}^{0}(w,z_{ij})) + \hat Y, z_{ik})}{b}\right]^{-\gamma_{i}} \\
		&+ \sum_{h=i+1}^{N} \bar{p}_{ih} \sum_{k=1}^{M} \tilde{p}_{jk} 
		\EE_{z_{ij}, z_{hk}}
		\hat \beta \hat R \frac{{\psi}_{h} b^{-\gamma_{i}}}{{\psi}_{i}} \left[\frac{c_n^{0} (\hat R (w-c_{n+1}^{0}(w,z_{ij})) + \hat Y, z_{hk})}{b}\right]^{-\gamma_{i}} \\
		&\leq c_{n+1}^0(w,z_{ij})^{-\gamma_{i}},
	\end{align*}
	where first inequality follows from \eqref{eq:cnc0}, and the second inequality follows from \eqref{eq:cn_lb} and the fact that $b^{-\gamma_h}/\psi_h \leq b^{-\gamma_i} /\psi_i$. Hence $c_{n+1}(w,z_{ij}) \geq c_{n+1}^0(w,z_{ij})$. This is a contradiction. The stated claim then follows by induction.
\end{proof}

\begin{proof}[Proof of Theorem \ref{t:non0_ampc3}]
	By Proposition~\ref{pr:F_fpt}, $F_i$ has a unique fixed point $x_i^* \in \RR_+^M$. By Lemma \ref{lm:limsup}, we have
	\begin{equation}
		\label{eq:non0_ub}
		\limsup_{w \to \infty} \frac{c^{*}(w,z_{ij})}{w} \leq x_{i}^{*}(\tilde z_{j})^{-1/\gamma_{i}}.
	\end{equation}
	Recall that, by assumption, $Y \geq b$ for some $b>0$ with probability one. We define
	\begin{equation}\label{eq:psih}
		\psi_h = \begin{cases}
			1, & \text{if } h < i, \\
			\alpha b^{\gamma_i - \gamma_h},
			& \text{if } h \geq i,
		\end{cases}
	\end{equation}
	where $\alpha$ can be arbitrarily small.
	By doing this, we can always choose $\psi$ such that the condition $r(U_i) < 1$ in Lemma~\ref{lm:Clim} holds. By Lemma~\ref{lm:Clim}, $J_i$ has a unique fixed point $y_{i}^{*}\in \RR_+^M$. 
	
	Moreover, note that by the definition of $\psi_h$ in \eqref{eq:psih}, we have $\psi_i = \alpha$ and $\psi_h = \psi_i b^{\gamma_i-\gamma_h}$ for $h > i$. Hence, $b^{-\gamma_h} / \psi_h = b^{-\gamma_i} / \psi_i$. By the parametric monotoniciy result from Proposition~\ref{pr:pm}, we have
	\begin{equation}
		\label{eq:c_para_mono}
		c^{*}(w,z_{ij}) \geq c^{*}(w,z_{ij};\psi, \gamma_{i}).
	\end{equation}
	Define the sequence $\{x_{n}\} \subset \RR_{+}^{M}$ by $x_{0}=y_{i}^{*}$ and 
	$x_{n}=F_{i}x_{n-1}$. Then, by Lemma~\ref{lm:Clim}, 
	\begin{equation*}
		\lim_{w \to \infty} \frac{c^{*}(w,z_{ij};\psi,\gamma_i)}{w} 
		= x_0(\tilde z_j)^{-1/\gamma_i}.
	\end{equation*}
	Since by assumption $R \geq m$ for some $m>0$, Lemma~\ref{lm:lim} applies. 
	Applying $T^{n}$ to both sides of (\ref{eq:c_para_mono}) and using Lemma~\ref{lm:lim} repeatedly, we obtain
	\begin{equation*}
		\liminf_{w \to \infty} \frac{T^n c^{*}(w,z_{ij})}{w} 
		\geq \lim_{w \to \infty} \frac{T^n c^{*}(w,z_{ij};\psi,\gamma_i)}{w}
		= x_{n}(\tilde z_{j})^{-1/\gamma_{i}}.
	\end{equation*}
	Letting $n\to \infty$, since $x_{n}(\tilde z_{j}) \to x_{i}^{*}(\tilde z_{j})$ and $c^* = T^n c^*$, we obtain
	\begin{equation}
		\label{eq:non0_lb}
		\liminf_{w \to \infty} \frac{c^{*}(w,z_{ij})}{w} \geq x_{i}^{*}(\tilde z_{j})^{-1/\gamma_{i}}.
	\end{equation}
	Then the stated claim follows from \eqref{eq:non0_ub} and \eqref{eq:non0_lb}.
\end{proof}

\section*{Acknowledgements}

Qingyin Ma gratefully acknowledges the financial support from NSFC (No.72573111).

\bibliographystyle{ecta}
\bibliography{os}

\begin{center}
	\LARGE 
	Online Appendix to ``A Theory of Saving under Risk Preference Dynamics''
\end{center}

In this online appendix, we present the quantitative algorithms used in Section~\ref{s:quant}. We first describe the endogenous grid method for solving the optimal consumption policy, followed by the algorithms for computing state-dependent impulse responses.

\section{The Endogenous Grid Algorithm}\label{s:EGA}

The model is solved using a generalized endogenous grid method \citep{carroll2006method}, which avoids root-finding during operator iterations by fixing a grid for saving rather than wealth. Specifically, we construct a finite saving grid $\mathcal{S}_{G}:=\{s_{g}\}_{g=1}^{G}$, where $0=s_{1}<\cdots<s_{G}$, update a candidate consumption function $c$ on $\mathcal{S}_{G} \times \ZZ$, and then determine the corresponding wealth grid endogenously based on optimal consumption and saving.

For each $t \in \NN$, let $\mathcal{W}^{t}_{G}:=\left\{w^{t}_{g}(z)\right\}_{(g,z)=(1,1)}^{(G,Z)}$ denote the endogenous wealth grid at iteration $t$, where $w^{t}_{g}(z)$ is the wealth corresponding to saving $s_{g}$ and exogenous state $z$. Set $w^{0}_{g}(z)=s_{g}$ for all $g$ and $z$. Let $\cC(\mathcal{W}^{t}_{G})$ denote the set of continuous piecewise linear functions $c:(0,\infty)\times \ZZ \to \RR$ such that for each $z \in \ZZ$: (1) $0<c(w,z)\leq w$ for all $w>0$; (2) $c(w,z)=w$ for $0< w \leq w^{t}_{1}(z)$; (3) $c(w,z)$ is linear in $w$ on each subinterval $[w^{t}_{g}(z),w^{t}_{g+1}(z)]$ for $g=1,...,G-1$; and (4) $c(w,z)$ is linearly extrapolated for $w>w^{t}_{G}(z)$. 

\begin{algorithm}
	\caption{\, The endogenous grid algorithm}\label{alg:egm}
	\small 
	\begin{enumerate}[label=\textbf{Step \arabic*.}, leftmargin=*, ref=\arabic*, itemsep=1mm]
		\item(Initialization). Choose a convergence criterion $\varpi> 0$, a saving grid $\left\{s_{g}\right\}$, and an initial policy  $c_{0}(w,z)$.
		\item\label{item:ega_update} (Updating). For each $t \in \NN$, given $c_{t-1} \in \cC(\mathcal{W}^{t-1}_{G})$, update it as follows:
		\begin{enumerate}
			\item For each $(s_g,z) \in \mathcal S_G \times \ZZ$, compute $\hat w:=\hat R s_{g}+\hat Y$ and define
			\begin{equation*}
				\tilde c_{t}(s_g,z):=u'(\cdot,z)^{-1}\left[ 
				\EE_z \hat \beta \hat R u'\left( 
				c_{t-1}\big(\hat w, \hat Z\big), \hat Z
				\right)
				\right].
			\end{equation*}
			%
			
			\item Define the upated wealth grid $\mathcal{W}^{t}_{G}:=\left\{w^{t}_{g}(z)\right\}_{(g,z)=(1,1)}^{(G,Z)}$ and the optimal consumption on it by
			\begin{equation*}
				w^{t}_{g}(z):=s_{g}+\tilde c_{t}(s_{g},z) \quad \text{and} \quad c_{t}(w^{t}_{g}(z),z):=\tilde c_{t}(s_{g},z).
			\end{equation*}
			
			\item Extend $c_t$ to $c_t \in \cC(\mathcal{W}^{t}_{G})$ by linear interpolation and extrapolation and set $c_{t}(w,z)=w$ for $w\leq w^{t}_{1}(z)$.
		\end{enumerate}
		
		\item (Convergence). Repeat Step~\ref{item:ega_update} until  $\left\{c_{t}(w^{t}_{g}(z),z)\right\}_{(g,z)=(1,1)}^{(G,Z)}$ converge.
	\end{enumerate}
\end{algorithm}

The endogenous grid algorithm for computing the consumption function is summarized by Algorithm~\ref{alg:egm}. In our quantitative experiments, we use an exponential savings grid with minimum, median, and maximum values of $0$, $50$, and $10^5$, respectively, and 1000 grid points. The method for constructing the exponential grid follows Appendix~B of \cite{ma2022asymptotic}. Iterations are terminated at precision $\varpi=10^{-4}$.

\section{The Generalized Impulse Response Functions}\label{s:gen_IRFs}

Suppose a shock to the exogenous state process $\{Z_t\}$ occurs at the beginning of period $t$. As noted earlier, we follow \cite{koop1996impulse} and compute impulse response functions (IRFs) as state-dependent random variables, conditional on different initial states $(w_{t-1},Z_{t-1})$. To proceed, we define
\begin{equation*}
	F(w,z,\hat Z, \hat \epsilon) := R(z,\hat Z, \hat \epsilon)[w - c^*(w,z)] + Y(z,\hat Z, \hat \epsilon).
\end{equation*}
Algorithm~\ref{alg:girf} summarizes the computation of generalized IRFs, where we use the absolute consumption response for demonstration. The method can be easily adapted to compute percentage changes, as well as impulse responses for saving rates and consumption volatility.

\begin{algorithm}
	\caption{\,The generalized impulse response function}\label{alg:girf}
	\small 
	\begin{enumerate}[label=\textbf{Step \arabic*.},  leftmargin=*, ref=\arabic*, itemsep=1mm]
		\item (Initialization). Choose initial values $(w_{t-1},Z_{t-1})$, the horizon $H$, and the number of Monte Carlo samples $K$. Set
		\begin{equation*}
			(\check w^k_{t-1}, \check Z^k_{t-1})
			= (w_{t-1}^k, Z^k_{t-1})
			\equiv (w_{t-1}, Z_{t-1})
			\quad \text{for $k=1,\dots, K$.}
		\end{equation*}
		\item Randomly sample $(H+1) \times K$ values of income shocks $\{\epsilon^k_{t+h}\}_{(h,k)=(0,1)}^{(H,K)}$.
		\item (Baseline Economy). Sample $(H+1)\times K$ values of the exogenous state 
		\begin{equation*}
			\left\{Z^k_{t+h}\right\}_{(h,k)=(0,1)}^{(H,K)}
			\quad \text{where } \; Z^k_{t+h} \sim P(Z^k_{t+h-1},\cdot\,).
		\end{equation*}
		
		\item\label{item:girf_next_state} (Impulse Shock Economy). Compute the period-$t$ exogenous state $\{\check Z^k_t\}_{k=1}^{K}$ after the shock. Sample $H \times K$ values of risk aversion
		\begin{equation*}
			\left\{\check Z^k_{t+h}\right\}_{(h,k)=(1,1)}^{(H,K)}
			\quad \text{where } \; \check Z^{k}_{t+h} \sim P(\check Z^k_{t+h-1},\cdot\,).
		\end{equation*}
		\item For $h=0,...,H$ and $k=1,...,K$, compute the sequence of wealth
		\begin{align*}
			&w_{t+h}^k = F(w_{t+h-1}^k, Z_{t+h-1}^k, Z_{t+h}^k, \epsilon_{t+h}^k) \\
			&\check w_{t+h}^k = F(\check w_{t+h-1}^k, \check Z_{t+h-1}^k, \check Z_{t+h}^k, \epsilon_{t+h}^k).
		\end{align*}
		\item For $h=0,...,H$, compute the period-$(t+h)$ impulse response
		\begin{equation*}
			\text{IRF}(t+h)
			= \frac{1}{K}\sum_{k=1}^{K}
			c^*(\check w_t^k, \check Z_t^k)
			- \frac{1}{K}\sum_{k=1}^{K}
			c^*(w_t^k, Z_t^k).
		\end{equation*}
	\end{enumerate}
\end{algorithm}

Recall the transition matrices $\bar P$ for $\{Z_t\}$ and $\tilde P$ for $\{\tilde Z_t\}$. In the quantitative experiments of Section~\ref{ss:quant_irf}, we have
\begin{equation*}
	Z_t = (\bar Z_t, \tilde Z_t) = (\gamma_t, \eta_t), \quad 
	\epsilon_t = (\nu_t, \epsilon_t^\eta)
	\quad \text{and} \quad  
	P(Z_{t-1}, Z_t) = \bar P(\gamma_{t-1}, \gamma_t) \tilde P (\eta_{t-1}, \eta_t).
\end{equation*}
A one-unit shock is applied to $\gamma_t$ at the beginning of period $t$. We discretize $\{\gamma_t\}$ into a 121-state Markov chain and $\{\eta_t\}$ into a 5-state Markov chain using the method of \cite{tauchen1986finite}. To compute the period-$t$ exogenous state in the impulse shock economy (Step~\ref{item:girf_next_state} of Algorithm~\ref{alg:girf}), we search the $\gamma_t$ states and identify the smallest $\gamma_t$ that is greater than or equal to $\mu_\gamma(1-\rho_\gamma) + \rho_\gamma \gamma_{t-1} + 1$. Given the fine discretization of $\{\gamma_t\}$, the resulting error is negligible.

\end{document}